\documentclass{article}

\usepackage{amsmath,amsthm,amsfonts}
\usepackage[ruled]{algorithm}
\usepackage[noend]{algorithmic}
\usepackage{eucal}
\usepackage{graphics}
\usepackage[colorlinks=true,citecolor=blue]{hyperref}
\usepackage{mathtools}
\usepackage{tikz}
\usepackage{multirow}

\newtheorem{thm}{Theorem}[section]

\newtheorem{lem}[thm]{Lemma}
\newtheorem{prop}[thm]{Proposition}
\theoremstyle{definition}

\theoremstyle{remark}

\newtheorem{prob}{Problem}

\def\NN{\mathbb{N}}

\newcommand{\abs}[1]{\left\vert#1\right\vert}
\newcommand{\pont}[3]{{#1}_{{#2} \rightarrow {#3}}}
\newcommand{\costA}[1]{v{\mathrm{A}_{#1}}}
\newcommand{\costN}[1]{v{\mathrm{N}_{#1}}}
\newcommand{\costSP}[1]{v{\mathrm{SP}_{\!#1}}}
\newcommand{\costSPN}[2]{{\mathrm{V}^{#2}_{\!#1}}}
\newcommand{\opt}{\protect\ensuremath\mathrm{opt}}
\newcommand{\Opt}{\protect\ensuremath\mathrm{OPT}}
\newcommand{\optStar}{\protect\ensuremath\mathrm{optStar}}
\newcommand{\distanceA}[1]{\opt{\mathrm{A}_{#1}}}
\newcommand{\distanceN}[1]{\opt{\mathrm{N}_{#1}}}
\newcommand{\DistanceSP}[1]{\opt{\mathrm{SP}_{\!#1}}}
\newcommand{\DistanceNSP}[2]{\opt{\mathrm{NSP}^{#2}_{\!#1}}}
\newcommand{\DistanceV}[2]{\opt{\mathrm{V}^{#2}_{\!#1}}} %LR incluiu aqui
\newcommand{\A}{\mathtt{a}}
\newcommand{\B}{\mathtt{b}}
\newcommand{\C}{\mathtt{c}}
\newcommand{\D}{\mathtt{d}}
\newcommand{\E}{\mathtt{e}}
\newcommand{\MSA}{\mathbf{MSA}}

\newcommand{\NMSA}{\mathbf{NMSA}}

\newcommand{\gap}{\protect\ensuremath\texttt{-}}
\newcommand{\probl}[1]{\protect\ensuremath\mathbf{#1}}

\newcommand{\sol}{\protect\ensuremath\mathrm{Sol}}
\newcommand{\vet}[1]{\vec{#1}}
\newcommand{\cjtoIndex}[1]{V_{#1}}
\newcommand{\conjuntoBool}[1]{\mathcal{B}^{{#1}}}
\newcommand{\metricA}{\mathbb{M}^\mathrm{W}}
\newcommand{\metricC}{\mathbb{M}^\mathrm{C}}
\newcommand{\metricN}{\mathbb{M}^\mathrm{N}}
\newcommand{\generalcost}{v}
\newcommand{\Generalcost}{V}
\newcommand{\cost}{cStar}

\newcommand{\qmax}{Q_{\max}}

\providecommand{\keywords}[1]
{
  \small	
  \textbf{Keywords:} #1
}

\title{Algorithms for normalized multiple sequence alignments}
\author{Eloi Araujo$^{1,}\footnote{Corresponding author}\qquad$ Diego P.~Rubert$^1$ \\ Luiz Rozante$^2\qquad$ Fábio V.~Martinez$^1$ \medskip \\
  \small $^1$ Faculdade de Computação \\[-1ex] \small Universidade Federal de Mato Grosso do Sul \\[-1ex] \small Brazil \\
  \small $^2$ Centro de Matemática, Computação e Cognição \\[-1ex] \small Universidade Federal do ABC \\[-1ex] \small Brazil
}

\begin{document}
\maketitle

% \input{abstract.tex}
%TODO mandatory: add short abstract of the document
\begin{abstract}
Sequence alignment supports numerous tasks in bioinformatics, natural
language processing, pattern recognition, social sciences, and other
fields. While the alignment of two sequences may be performed swiftly
in many applications, the simultaneous alignment of multiple sequences
proved to be naturally more intricate. Although most multiple sequence
alignment (MSA) formulations are NP-hard, several approaches have been
developed, as they can outperform pairwise alignment methods or are
necessary for some applications.  Taking into account not only
similarities but also the lengths of the compared sequences
(i.e.~normalization) can provide better alignment results than both
unnormalized or post-normalized approaches. While some normalized
methods have been developed for pairwise sequence alignment, none have
been proposed for MSA. This work is a first effort towards the
development of normalized methods for MSA.  We discuss multiple
aspects of normalized multiple sequence alignment (NMSA). We define
three new criteria for computing normalized scores when aligning
multiple sequences, showing the NP-hardness and exact algorithms for
solving the NMSA using those criteria. In addition, we provide
approximation algorithms for MSA and NMSA for some classes of scoring
matrices.


\end{abstract}

\section{Introduction}\label{sec:intro}

Sequence alignment lies at the foundation of bioinformatics. Several
procedures rely on alignment methods for a range of distinct purposes,
such as detection of sequence homology, secondary structure
prediction, phylogenetic analysis, identification of conserved motifs
or genome assembly. On the other hand, alignment techniques have also
been reshaped and found applications in other fields, such as natural
language processing, pattern recognition, or social
sciences~\cite{AT2000,apostolico1997,BL2002,marzal1993}.

Given its range of applications in bioinformatics, extensive efforts
have been made to improve existing or developing novel methods for
sequence alignment. The simpler ones compare a pair of sequences in
polynomial time on their lengths, usually trying to find editing
operations (insertions, deletions, and substitutions of symbols) that
transform one sequence into another while maximizing or minimizing
some objective function called edit distance~\cite{HAR2009}. This
concept can naturally be generalized to align multiple
sequences~\cite{WLXZ2015}, adding another new layer of algorithmic
complexity, though. In this case, most multiple sequence alignment
(MSA) formulations lead to NP-hard
problems~\cite{Elias2006}. Nevertheless, a variety of methods suitable
for aligning multiple sequences have been developed, as they can
outperform pairwise alignment methods on tasks such as phylogenetic
inference~\cite{OR2006}, secondary structure prediction~\cite{CB1999}
or identification of conserved regions~\cite{Clustal2011}.

In order to overcome the cost of exact solutions, a number of MSA
heuristics have been developed in recent years, most of them using the
so-called progressive or iterative
methods~\cite{HTHI1995,Clustal2014,TPP1999,MCoffee2006}. Experimental
data suggest that the robustness and accuracy of heuristics can still be
improved, however~\cite{WLXZ2015}.

Most approaches for pairwise sequence alignment define edit distances
as absolute values, lacking some normalization that would result in
edit distances relative to the lengths of the sequences. However, some
applications may require sequence lengths to be taken into
account. For instance, a difference of one symbol between sequences of
length 5 is more significant than between sequences of length 1000. In
addition, experiments suggest that normalized edit distances can
provide better results than both unnormalized or post-normalized edit
distances~\cite{marzal1993}. While normalized edit distances have been
developed for pairwise sequence
alignment~\cite{arslan1999,marzal1993}, none have been proposed for
MSA to the best of our knowledge.

In this work, we propose exact and approximation algorithms for
normalized MSA (NMSA). This is a first step towards the development of
methods that take into account the lengths of sequences for computing
edit distances when multiple sequences are compared.
The remainder of this paper is organized as
follows. Section~\ref{sec:pre} introduces concepts related to sequence
alignment and presents normalized scores for NMSA, followed by the
complexity analysis of NMSA using those scores in
Section~\ref{sec:complexity}. Next, Sections~\ref{sec:exatos}
and~\ref{sec:approx} describe exact and approximation algorithms,
respectively. Section~\ref{sec:concl} closes the paper with the
conclusion and prospects for future work.

\section{Preliminaries}\label{sec:pre}

An \emph{alphabet} $\Sigma$ is a finite non-empty set of
\emph{symbols}. A finite sequence $s$ with $n$ symbols in $\Sigma$ is
seen as $s(1) \cdots s(n)$. We say that the \emph{length} of $s$,
denoted by $\abs{s}$, is $n$. The sequence $s(p) \cdots s(q)$,
with $1 \le p \le q \le n$, is denoted by $s(p\!:\!q)$. If $p >
q$, $s(p\!:\!q)$ is the \emph{empty} sequence which length is equal to zero
and it is denoted by $\varepsilon$. We denote the sequence resulting
from the concatenation of sequences $s$ and $t$ by $s \cdot t = st$.
A sequence of
$n$ symbols~$\A$ is denoted by~$\A^{n}$. A $k$-tuple $S$ over
$\Sigma^*$ is called a \emph{$k$-sequence} and we write $s_1, \ldots,
s_k$ to refer to $S$, where $s_i$ is the $i$-th sequence in $S$.
We denote $S = \emptyset$ if every sequence of $S$ is a empty sequence.

Let $\Sigma_{\gap} := \Sigma \cup \{\gap\}$, where $\gap \not\in \Sigma$
and the symbol $\gap$ is called a \emph{space}. Let $S = s_{1},
\ldots, s_{k}$ be a $k$-sequence. An \emph{alignment} of $S$ is a
$k$-tuple $A = [s'_{1}, \ldots, s'_{k}]$ over $\Sigma_{\gap}^{*}$,
where
\begin{enumerate}
\item[($a$)] each sequence $s'_{h}$ is obtained by inserting spaces in
  $s_{h}$;
  \item[($b$)] $\abs{s'_{h}} = \abs{s'_{i}}$ for each pair $h, i$, with
    $1 \le h, i \le k$; 
    \item[($c$)] there is no $j$ in $\{1, \ldots, k\}$ such
      that $s'_{1}(j) = \ldots = s'_{k}(j) = \gap$. 
\end{enumerate}
Notice that alignments, which are $k$-tuples
over $\Sigma_{\gap}^*$, are written enclosed by square brackets
``$[~]$''.
The sequence $A(j) = [s'_{1}(j), \ldots, s'_{k}(j)] \in \Sigma_{\gap}^k$ is the
\emph{column $j$} of the alignment $A = [s'_{1}, \ldots, s'_{k}]$ and
$A[j_{1} \!:\! j_{2}]$ is the alignment
defined by the columns $j_{1}, j_{1} + 1, \ldots, j_{2}$ of $A$.
We say that the pair $[s'_{h}(j), s'_{i}(j)]$ \emph{aligns} in
$A$ or, simply, that $s'_{h}(j)$ and $s'_{i}(j)$ are \emph{aligned} in
$A$, and $\abs{A} = \abs{s'_{i}}$ is the \emph{length} of the
alignment~$A$.
It is easy to check that $\max_{i} \{\abs{s_{i}} \} \le
\abs{A} \le \sum_{i} \abs{s_{i}}$.
We denote by $\mathcal{A}_{S}$ the set of all alignments of $S$.

An alignment can be used to represent \emph{editing operations} of
\emph{insertions}, \emph{deletions} and \emph{substitutions} of
symbols in sequences, where the symbol $\gap$ represents insertions or
deletions. An alignment can also be represented in the matrix format
by placing one sequence above another.
Thus, the alignments
\begin {eqnarray*}
  [\A \A \A \gap, \A \B \gap \gap, \gap \C \A \C]
  & \mbox{and} &
  [\gap \A \A \A \gap, \A \B \gap \gap \gap, \gap \C \A \gap \C]
\end {eqnarray*}
of 3-sequence $\A \A \A, \A \B, \C \A \C$ can be represented respectively as
\begin {eqnarray*}
\left[
\begin {array} {cccc}
\A & \A & \A & \gap \\
\A & \B & \gap & \gap \\
\gap & \C & \A & \C
\end {array}
\right]
& \mbox{and} &
\left[
\begin {array} {ccccc}
\gap & \A & \A & \A & \gap \\
\A & \B & \gap & \gap & \gap \\
\gap & \C & \A & \gap & \C
\end {array}
\right].
\end {eqnarray*}

Let $I = \{i_{1}, \ldots, i_{m}\} \subseteq \{ 1, \ldots, k \}$ be a
set of indices such that $ i_{1} < \cdots < i_{m}$ and let $A =
[s'_{1}, \ldots, s'_{k}]$ be an alignment of a $k$-sequence $S = s_{1}, \ldots,
s_{k}$. We write $S_I$ to denote the $m$-sequence $s_{i_{1}}, \ldots,
s_{i_{m}}$. The alignment of $S_I$ \emph{induced} by $A$ is the
alignment $A_{I}$ obtained from the alignment $A$, considering only
the corresponding sequences in $S_I$ and, from the resulting
structure, removing columns where all symbols are $\gap$.
In the
previous example, being $A = [\A \A \A \gap, \A \B \gap \gap, \gap \C \A
  \C]$ an alignment of $\A \A \A, \A \B, \C \A \C$, we have 
\[
\left[
  \begin{array}{ccc}
    \A & \A & \A \\
    \A & \B & \gap
  \end {array}
  \right]
\]
is an alignment of $\A \A \A, \A \B$ induced by $A$.

A \emph{$k$-vector} $\vec{\jmath} = [j_1, \ldots, j_k]$
is a $k$-tuple, where
$j_i \in
\mathbb{N} = \{0, 1, 2, \ldots\}$. We say that $j_i$ is the $i$-th element of
$\vec{\jmath}$. The $k$-vector $\vec{0}$ is such that all its elements
are zero. If $\vec{\jmath}$ and $\vec{h}$ are $k$-vectors, we write
$\vec{\jmath} \le \vec{h}$ if $j_i \le h_i$ for each $i$; and
$\vec{\jmath} < \vec{h}$ if $\vec{\jmath} \le \vec{h}$ and
$\vec{\jmath} \not= \vec{h}$.
A sequence of $k$-vectors
$\vec{\jmath}_1, \vec{\jmath}_2, \ldots$ is in \emph{topological
  order} if $\vec{\jmath}_i < \vec{\jmath}_{h}$ implies $i < h$.
Given two $k$-vectors, say
$\vec{a} = [a_1,a_2 \ldots a_k]$
$\vec{b} = [b_1,b_2 \ldots b_k]$, we
say that $\vec{a}$ \emph{precedes} $\vec{b}$
when there exists $l \in \mathbb{N}$ such that $a_i = b_i$ for each $i>l$ and
$a_l < b_l$.
A sequence of $k$-vectors
$\vec{\jmath}_1, \vec{\jmath}_2, \ldots$ is in \emph{lexicographical order}
if $\vec{\jmath}_i$ precedes $\vec{\jmath}_{h}$ for each $i < h$.
Clearly, a lexicographical order is a special case of topological order.

Consider $S = s_1, \ldots, s_k$ a $k$-sequence with $n_i = \abs{s_i}$
for each $i$ and we call $\vec{n} = [n_1, \ldots, n_k]$ the \emph{length of $S$}. Let
$\cjtoIndex{\vec{n}} = \{\vec{\jmath} : \vec{\jmath} \le \vec{n} \}$.
%% For
%% example, if $k= 3$ and $\vec{n} = [1, 2, 1]$, then
%% $\cjtoIndex{\vec{n}} = \{ [x, y, z]: x, y, z \in \mathbb{N}, x \le 1,
%% y \le 2, z \le 1 \}$.
Therefore, 
$\abs{\cjtoIndex {\vec{n}}} = \prod_i (\abs{s_i} + 1)$ which implies that
if $n_{i} = n$ for all $i$, then
$\abs{\cjtoIndex {\vec{n}}} = (n+1)^{k}$. Define $S(\vec{\jmath}) = s_1(j_1),
\ldots, s_k(j_k)$ a \emph{column} $\vec{\jmath}$ in $S$ and we say that
$S(1\!:\!\vec{\jmath}) = s_{1}(1\!:\!j_{1}), \ldots, s_{k}(1\!:\!j_{k})$ is
the \emph{prefix of $S$ ending in $\vec{\jmath}$}. Thus, $S =
S(1\!:\!\vec{n})$. Besides that, if $A$ is an alignment and $\vec {j}
= [j, j, \ldots, j]$, then $A[\vec {j}] = A(j)$.

Denote by $\conjuntoBool{k}$ the set of $k$-vectors $[b_1, \ldots,
  b_k]$, where $b_i \in \{0, 1\}$ for each $i$.
Now, for $\vec{b} \le \vec{j}$, where $\vec{j} = [j_1, \ldots, j_k]$, define
\[
\vec{b} \cdot S(\vec{j}) = [x_{1}, \ldots, x_{k}] \in
\Sigma_{\gap}^{k}\,,
\]
such that
\[
x_{i} = \left\{ \begin{array}{ll}
  s_{i}(j_{i})\,, & \mbox{if $b_{i} = 1$} \\
  \gap\,, & \mbox{otherwise}\,.
\end{array} \right.
\]
%and we also say that \emph{$\vec{b}$ is associated with column $\vec{j}$}.
Therefore, given an alignment $A$ of
$S(1\!:\!\vec{n})$, there exists $\vec{b} \in \conjuntoBool{k}$, with
$\vec{b} \le \vec{n}$, such that $A(\abs{A}) = \vec{b} \cdot
S(\vec{n})$.
In this case, notice that $b_{i} = 1$ if and only if
$s_i(n_i)$ is in the $i$-th row of the last column of $A$ and we
say that $\vec{b}$ \emph{defines} the column $\abs{A}$ of alignment $A$.
For $\vec{b} \le \vec{\jmath}$, we also
define the operation
\[
\vec{\jmath} - \vec{b} = [\jmath_{1} - b_{1}, \ldots, \jmath_{k} - b_{k}]\,.
\]
Notice that $\abs{\mathcal{B}_{k}} = 2^{k}$.
Figure~\ref{fig:example1} shows an example of alignment using vectors
to define columns and operations.

\begin{figure}[ht]
  \begin{eqnarray*}
  A = \left[
    \begin{array}{ccccc}
      \A & \B & \C & \gap & \gap\\
      \gap & \B & \C & \A & \gap\\
      \B & \gap & \B & \gap & \A\\
      \A & \A & \A & \A & \A\\
      \C & \gap & \gap & \gap & \gap
      \end{array}
    \right]
  & \Rightarrow &
  \left[
    \begin{array}{c}
      \\
      \\
      A[1:4] \\
        \\ \\
      \end{array}
    \right]
  \left[
    \begin{array}{c}
      \\
      \\
      \left[ 0,0,1,1,0 \right]  \cdot
      S(1:[3,3,3,5,1]) \\
      \\ \\
      \end{array}
    \right]
  \end{eqnarray*}
  \caption{The length of alignment $A$ of 
    $S = \A\B\C, \B\C\A, \B\B\A, \A\A\A\A\A, \C$ is 
    $\vec{n} = [\abs{\A\B\C}, \abs{\B\C\A}, \abs{\B\B\A}, \abs{\A\A\A\A\A}, \abs{\C}]= [3,3,3,5,1]$.
    Here, $S(\vec{n}) = [\C, \A, \A, \A, \C]$ and since 
    $\vec{b_5} = [0, 0, 1, 1, 0]$ define the last column of $A$, we have
    $S (\vec{n} - \vec{b_5}) = \A\B\C, \B\C\A, \B\B, \A\A\A\A, \C$.
    The alignment defined by the first four columns of $A$ is $A[1:4]$
    which is an alignment of $S(1:\vec{n} - \vec{b_5})$ and the last column of
    $A$ is $\vec{b_5} \cdot S(\vec{n})$.
    It is interesting to note
    that $A$ can be completely defined by $5$-vectors
    $\vec{b_1} = [1,0,1,1,1],\vec{b_2} = [1,1,0,1,0], \vec{b_3} = [1,1,1,1,0],\vec{b_4} = [0,1,0,1,0],\vec{b_5} = [0,0,1,1,0]$ where $b_i$ represents column $i$of alignment $A$.
    }\label{fig:example1}
  \end{figure}

For a problem $\probl{P}$, we call $\mathbb{I}_{\probl{P}}$ the set of
instances of $\probl{P}$. If $\probl{P}$ is a decision problem, then
$\probl{P}(I) \in \{\texttt{Yes}, \texttt{No}\}$ is the image of an
instance $I$. If $\probl{P}$ is an optimization (minimization)
problem, there is a set $\sol(I)$ for each instance $I$, a function
$v$ defining a non-negative rational number for each $X \in \sol(I)$,
and a function $\opt_v(I) = \min_{X \in \sol(I)} \{v(X)\}$. We use
$\opt$ instead of $\opt_v$ if $v$ is obvious.
Let $\mathbf{A}(I) = v(X \in  \sol(I))$ a 
solution computed by an algorithm $\mathbf{A}$ with input~$I$.
We say that $\mathbf{A}$ is an
\emph{$\alpha$-approximation} for $\probl{P}$ if $\mathbf{A}(I) \le
\alpha \, \opt(I)$ for each $I \in \mathbb{I}_{\probl{P}}$.
We say that $\alpha$ is an \emph{approximation factor}
for~$\probl{P}$.

The \emph{alignment problem} is a collection of decision and
optimization problems whose instances are finite subsets of $\Sigma^*$
and $\sol(S) = \mathcal{A}_S$ for each instance $S$. Function $v$,
used for scoring alignments, is called \emph{criterion} for
$\probl{P}$ and we call $v[A]$ the \emph{cost of the alignment} $A$.
The \emph{$v$-optimal alignment} $A$ of $S$ is such that $v[A] =
\opt(S)$. Thus, we state the following general optimization problems
using the criterion $v$ and integer $k$:
\begin{prob}[Alignment with criterion $v$]
  Given a $k$-sequence $S$, find the cost of a
  $v$-optimal alignment of~$S$.
\end{prob}
We also need the decision version of the alignment problem with
criterion $v$ where we are given a $k$-sequence $S$ and a number $d
\in \mathbb{Q}_\ge$, and we want to decide whether there exists an
alignment $A$ of $S$ such that $v[A] \le d$.

Usually the cost of an alignment $v$ is defined from a scoring
matrix. A \emph{scoring matrix} $\gamma$ is a matrix whose
elements are rational numbers.
The matrix has rows and columns indexed by symbols in $\Sigma_{\gap}$.
For $\A, \B \in \Sigma_{\gap}$ and a scoring matrix $\gamma$,
we denote by $\pont{\gamma}{\A}{\B}$ the entry of $\gamma$ in row
$\A$ and column $\B$. The value $\pont{\gamma}{\A}{\B}$ defines the score
for a substitution if $\A, \B \in \Sigma$, for an insertion if $\A =
\gap$, and for a deletion if $\B = \gap$. The entry
$\pont{\gamma}{\gap}{\gap}$ is not defined.

\subsection{$\costA{\gamma}$- and $\costN{\gamma}$-score for scoring alignments of 2-sequences}

Consider a scoring matrix $\gamma$. Let $S = s, t \in \Sigma^{*}$ be a 2-sequence whose length is $\vec{n} = [n
= \abs{s}, m = \abs{t}]$. A simple criterion for scoring alignments
of 2-sequences using the function $\costA{\gamma}$ follows. For an alignment
$[s', t']$ of $s, t$ we define
\[
\textstyle \costA{\gamma}[s', t'] = \sum_ {j = 1}^{|[s', t']|}
\pont{\gamma}{s'\!(j)}{t'\!(j)}\,.
\]
We say that $\costA{\gamma}[s', t']$ is a
$\costA{\gamma}$-\emph{score} of $s, t$. The optimal function for this
criterion is
called \emph{unweighted edit distance} and
denoted by $\distanceA{\gamma}$, and
an $\distanceA{\gamma}$-optimal alignment is also called 
an A-\emph{optimal alignment}.
When $\pont{\gamma}{\A}{\A} = 0$ and $\pont{\gamma}{\A}{\B} = 1$
for each $\A \not= \B \in \Sigma_{\gap}$, $\distanceA{\gamma}$ is also known as
\emph{Levenshtein distance}~\cite{levenshtein1966}.

Now, suppose that $n \ge m$. Needleman and Wunch~\cite{needleman1970}
proposed an $O(n^2)$-time algorithm for computing
\[
\distanceA{\gamma}(s, t) = \distanceA{\gamma}(S) =
\min_{A \in \mathcal{A}_S} \costA{\gamma} [A] =
\left\{
\begin{array}{ll}
  0 & \mbox{if $S = \emptyset$},\\
  \min_{\vec{b} \in \conjuntoBool{k}, \vec{b} \le \vec{n}} \Big\{ \distanceA {\gamma} (S[1:\vec{n} - \vec{b}) + \costA{\gamma} [\vec{b} \cdot S] \Big\} & \mbox{otherwise}.
  \end{array}
\right.
\]
If $\distanceA{\gamma}$ is a Levenshtein
distance, Masek and Paterson~\cite{masek1980} presented an $O(n^{2} /
\log n)$-time algorithm using the ``\emph{Four Russian's
  Method}''. Crochemore, Landau and Ziv-Ukelson~\cite{crochemore2002}
extended $\distanceA{\gamma}$-score supporting 
scoring matrices with real numbers and  describing an $O(n^{2} / \log
n)$-time algorithm. Indeed, there is no algorithm to determine
$\distanceA{\gamma}(s, t)$ in $O(n^{2 - \delta})$-time for any
$\delta> 0$, unless SETH is
false~\cite{DBLP:journals/siamcomp/BackursI18}. Andoni, Krauthgamer
and Onak~\cite{andoni2010polylogarithmic} described a nearly linear
time algorithm approximating the edit distance within an approximation
factor $\textrm{poly}(\log n)$. Later, Chakraborty et
al.~\cite{chakraborty2020approximating} presented an
$O(n^{2-2/7})$-time $\alpha$-approximation for edit distance, where
$\alpha$ is constant.

% \subsection{$\costN{\gamma}$-score for two sequences} \label{sec:A-opt}

Marzal and Vidal~\cite{marzal1993} defined another criterion for
scoring alignments of two sequences called
$\costN{\gamma}$-\emph{score}, which is a normalization of
$\costA{\gamma}$-score, as follows:
\[
\costN{\gamma}[A] =
\left\{
\begin{array}{ll}
  0\,, & \mbox{if $\abs{A} = 0$}\,, \\
  \costA{\gamma}[A] / \abs{A}\,, & \mbox{otherwise}\,.
\end{array}
\right.
\]
The optimal function for this criterion as known as
\emph{normalized edit distance} and it is denoted by $\distanceN{\gamma}$, and
an $\distanceN{\gamma}$-optimal alignment is also called 
a N-\emph{optimal alignment} of $S = s, t$.

A naive dynamic programming algorithm was proposed by Marzal and
Vidal~\cite{marzal1993} to obtain an N-optimal alignment of two
sequences in $O(n^3)$-time. Using fractional programming, Vidal,
Marzal and Aibar~\cite{vidal1995} presented an algorithm with running
time $O(n^3)$, but requiring only $O(n^2)$-time in practice which is
similarly to the
classical (unnormalized) edit distance algorithm. Further, Arslan and
Egecioglu~\cite{arslan1999} described an $O(n^2 \log n)$-time
algorithm to this problem.

Let $S = s,t$ a 2-sequence whose length is
$\vec{n}$.
A simple heuristic for determining a close value for $\distanceN{\gamma}(s,t)$ computes first $\distanceA{\gamma}(s,t)$  
and then divides it by the maximum length $L(\vec{n}) = L(S)$ of an optimum alignment of $S$.
In order to find $L(S)$, we use the same dynamic programming strategy for computing $\distanceA{\gamma}(s,t)$, i.e.,
\[
L(S) = L(\vec{n}) = \left\{
\begin{array}{ll}
  0 & \mbox{if $S = \emptyset$}\\
  \max_{\vec{b} \in \conjuntoBool{k}, \vec{b} \le \vec{n}, \distanceA{\gamma}(S) = \distanceA{\gamma}(S - \vec{b}) + \costA{\gamma} (\vec{b} \cdot S)}
    \Big\{ L(\vec{n} - \vec{b}) \Big\} + 1
       & \mbox{otherwise},
  \end{array}
\right.
\]
%% Let $A$ be an A-optimal alignment of maximum length of $2$-sequence $S
%% = s, t$, with $\abs{s}= n$ and $\abs{t} = m$. Considering $\vec{n} =
%% [n, m]$ and $\vec{b}$ a bit vector such that $A(\abs{A}) = \vec{b}
%% \cdot S(\vec{n})$, the length of a maximum length A-optimal alignment of
%% $S(1\!:\!\vec{n} - \vec{b})$ must be $\abs{A} - 1$. Thus, the maximum
%% length $L(n, m)$ can be found by a dynamic programming formula as
%% following:
%% \begin{align*}
%%   L(0,0) &= 0\,,\quad L(0,j) = j\,, \quad L(i,0) = i\,, \\
%%   L(i,j) &= 
%%   \max\left\{
%%   \begin{array}{ll}
%%     \!\! L(i-1,j)\,, & \mbox{if}~d(i,j) = d(i-1, j) +
%%     \pont{\gamma}{s(i)}{\gap} \\
%%     \!\! L(i,j-1)\,, & \mbox{if}~d(i,j) = d(i,j-1) +
%%     \pont{\gamma}{\gap}{t(j)} \\
%%     \!\! L(i-1, j-1)\,, & \mbox{if}~d(i,j) = d(i-1,j-1) +
%%     \pont{\gamma}{s(i)}{s(j)}
%%   \end{array}
%%   \!\! \right\}\!+\!1\,, \, i,j>0\,,
%% \end{align*}
%% where $d(i, j) = \distanceA{\gamma}(s(1\!:\!i), t(1\!:\!j))$.
% Therefore, the maximum length A-optimal alignment of $s, t$ can be
% obtained
in $O(nm)$-time. The following theorem shows that this heuristic is
a simple approximation algorithm to find a N-optimal
alignment.

\begin{thm} \label{thm:approx-1}
  Let $s, t$ be a 2-sequence of length $[n, m]$ and let
  $L(s,t)$ be the maximum length of an A-optimal alignment of $s, t$. Then,
  \[
  \frac{\distanceA{\gamma}(s,t)}{L(S)} \le 2 \, \distanceN{\gamma}(s,
  t)\,,
  \]
  and it can be computed in $O(n^2)$-time if $n = m$. Moreover, this
  ratio is tight, i.e., for any positive rational $\varepsilon$, there
  exists a scoring matrix $\gamma$, sequences $s, t$ and an
  $A$-optimal alignment of $s,t$ with maximum length $A$ such that
  \[
  \frac{\distanceA{\gamma}(s, t)}{\abs{A}} = \frac{\costA{\gamma}[A]}{\abs{A}} = (2 -
  \varepsilon)\,\distanceN{\gamma}(s, t).
  \]
\end{thm}

\begin{proof}
  Let $A$ be an $A$-optimal alignment with maximum length computed by
  the heuristic above in $O(nm)$-time. Let $B$ be an
  N-optimal alignment. Thus, $\costA{\gamma}[A] \le
  \costA{\gamma}[B]$. Moreover,
  since $\abs{B} \le n + m$, $n + m \le 2\,\max\{n, m
  \}$ and $\max\{n, m
  \} \le 2\abs{A}$, we have that $\abs{A} \ge \abs{B}/2$. Therefore,
  $\costN{\gamma}[A] = \frac{\costA{\gamma}[A]}{\abs{A}} \le
  \frac{\costA{\gamma}[B]}{\abs{A}} \le
  \frac{\costA{\gamma}[B]}{\abs{B}/2} = 2\,\distanceN{\gamma}(s,
  t)\,.$

  We present now a 2-sequence $s, t$ and a scoring matrix $\gamma$ such that
  the solution given by the heuristic is at least $(2 - \varepsilon) \distanceN{\gamma}(s,
  t)$ for any
  $\varepsilon$ in $\mathbb{Q}_>$. Let $\Sigma = \{ \A, \B \}$,
  $\gamma$ be a scoring matrix such that $\pont{\gamma}{\A}{\gap} =
  \pont{\gamma}{\B}{\gap} = 1/\varepsilon$ and $\pont{\gamma}{\A}{\B}
  = 2/\varepsilon - 1 $ and $\A^n, \B^n \in \Sigma^*$, where $n$ in
  a positive integer. Observe that the $\costA{\gamma}$-score of any
  alignment of $(\A^n, \B^n)$, where $[\A, \B]$ is aligned in $k$
  columns, is equal to $2n/\varepsilon - k$. Thus, $ \distanceA{\gamma}(\A^n,
  \B^n) = \min_{0 \le k \le n} \{ 2n/\varepsilon - k \} =
  2n/\varepsilon - n = (2/\varepsilon - 1) \,n$ which implies $[\A^n,
    \B^n]$ is the $A$-optimal alignment with maximum length.  Since
  $\distanceN{\gamma}(\A^n, \B^n) \le \costN{\gamma}([\A^n\gap^n,
    \gap^n\B^n]) = 1/\varepsilon$, it follows
  \begin{align*}
    \frac{\distanceA{\gamma}(\A^n, \B^n)}{\abs{[\A^n, \B^n]}}
    & = \frac{\costA{\gamma}[\A^n, \B^n]}{\abs{[\A^n, \B^n]}} =
    \frac{(2/\varepsilon - 1)\,n}{n} =
    \frac{2 - \varepsilon}{\varepsilon} =
    (2 - \varepsilon)\, \costN{\gamma}([\A^n\gap^n,
    \gap^n\B^n])
    \ge (2 - \varepsilon)\,\distanceN{\gamma}(\A^n, \B^n)\,.
  \end{align*}
  
\end{proof}

We define now classes of scoring matrices. The most common class of scoring
matrices $\metricC$ has the following properties: for all $\A,
\B, \C, \in \Sigma_ {\gap}$, we have (a) $\pont{\gamma}{\A}{\B} > 0$
if $\A \not= \B$, and $\pont{\gamma}{\A}{\B} = 0$ if $\A = \B$; (b)
$\pont{\gamma}{\A}{\B} = \pont{\gamma}{\B}{\A}$; and (c)
$\pont{\gamma}{\A}{\C} \le \pont{\gamma}{\A}{\B} +
\pont{\gamma}{\B}{\C}$. The class $\metricA$ of scoring matrices is
such that, for all symbols $\A, \B, \C \in \Sigma$, we have (a)
$\pont{\gamma}{\A}{\gap} = \pont{\gamma}{\gap}{\A} > 0$; (b)
$\pont{\gamma}{\A}{\B} > 0$ if $\A \not= \B$, and
$\pont{\gamma}{\A}{\B} = 0$ if $\A = \B$; (c) if
$\pont{\gamma}{\A}{\B} < \pont{\gamma}{\A}{\gap} +
\pont{\gamma}{\gap}{\B}$, then $\pont{\gamma}{\A}{\B} =
\pont{\gamma}{\B}{\A}$; (d) $\pont{\gamma}{\A}{\gap} \le
\pont{\gamma}{\A}{\B} + \pont{\gamma}{\B}{\gap}$; and (e) $\min\{
\pont{\gamma}{\A}{\C}, \pont{\gamma}{\A}{\gap} +
\pont{\gamma}{\gap}{\C} \} \le \pont{\gamma}{\A}{\B} +
\pont{\gamma}{\B}{\C}$. Moreover, the class $\metricN$ is such that
(a) $\metricN \subseteq \metricA$ and (b) $\pont{\gamma}{\A}{\gap} \le
2\,\pont{\gamma}{\B}{\gap}$ for each $\A, \B \in \Sigma$.

For a set $S$, we say that $f: S \times S \rightarrow \mathbb{R}$ is a
\emph{distance function} or \emph{metric on $S$} if $f$ satisfies for all $s, t, u
\in S$,:  
\begin{enumerate}
  \item $f(s, s) = 0$
    (\emph{reflexivity});
  \item $f(s, t) > 0$ if $s \not= t$
    (\emph{strict positiveness});
  \item $f(s, t) = f(t, s)$ (\emph{symmetry}); and
    \item $f(s, u) \le f(s, t) + f(t, u)$ (\emph{triangle inequality}).
\end{enumerate}

If a given criterion $v$ depends on a scoring matrix $\gamma$ and it is a
metric on $\Sigma^*$, we say that the scoring matrix $\gamma$
\emph{induces} a $\opt_v$-distance on $\Sigma^*$.
Sellers~\cite{sellers1974} showed that
matrices in~$\metricC$ induce an $\distanceA{\gamma}$-distance on
$\Sigma^*$.
Araujo and Soares~\cite{araujo2006} showed that $\gamma
\in \metricA$ and $\gamma \in
\metricN$ if and only if $\gamma$ induces an
$\distanceA{\gamma}$-distance and an $\distanceN{\gamma}$-distance on $\Sigma^*$
respectively. Figure~\ref{fig:sets1}
shows the relationship between these classes.

\begin{figure}[hbt]
  \begin{center}
    \begin{tikzpicture}
      \node
      % [font=\fontsize{18}{18}\selectfont] 
      at (0, .8) {$\metricA$};
      \node
      % [font=\fontsize{18}{18}\selectfont] 
      at (-.8, 0) {$\metricC$};
      \node
      % [font=\fontsize{18}{18}\selectfont] 
      at (.9, 0) {$\metricN$};
      
      % clauses 
      \draw (0,0) ellipse (2.5cm and 1.2cm);
      \draw (-.8,0) ellipse (1cm and .6cm);
      \draw (.8,0) ellipse (1cm and .6cm);
    \end{tikzpicture}
  \end{center}
  \vspace*{-.4cm}
  \caption{Relationship between scoring matrices. Araujo and
    Soares~\cite{araujo2006} showed that $\metricC \subseteq
    \metricA$, $\metricN \subseteq \metricA$, $\metricC \not\subseteq
    \metricN$ and $\metricN \not\subseteq \metricC$. Moreover, the
    scoring matrix $\gamma$ such that $\pont{\gamma}{\A}{\A} = 0$ for
    each $\A$ and $\pont{\gamma}{\A}{\B} = 1$ for each $\A \not= \B$
    is in $\metricC \cap \metricN$, which implies that $\metricC \cap
    \metricN \not= \emptyset$.}\label{fig:sets1}
\end{figure}

Given a scoring function $\generalcost$ for alignments
of 2-sequences $s, t$ that
depends on a scoring matrix, we say that two scoring matrices
$\gamma$ and $\rho$ are \emph{equivalent} considering $\generalcost$
when
$\generalcost_\gamma[A] \le \generalcost_\gamma[B]$ if and only if
$\generalcost_\rho[A] \le \generalcost_\rho[B]$ for any pair of
alignments $A, B$ of sequences $s, t$. If $\rho$ is a matrix obtained
from $\gamma$ by multiplying each entry of $\gamma$ by a constant $c >
0$, then $\costA{\rho}[A] = c \cdot \costA{\gamma}[A]$ and
$\costN{\rho}[A] = c \cdot \costN{\gamma}[A]$, which implies that
$\gamma$ and $\rho$ are equivalent. As a consequence, when
the scoring function is $\costA{\gamma}$ or $\costN{\gamma}$
and it is convenient,
we can suppose that all entries of $\gamma$ are integers instead of
rationals.

\subsection{$\costSP{\gamma}$-score for $k$ sequences}

Consider a scoring matrix $\gamma$. Let $S = s_{1}, \ldots, s_{k}$ be
a $k$-sequence whose length is $\vec{n}$
and $A = [s'_{1}, \ldots, s'_{k}]$ be an alignment of
$S$. The criterion $\costSP{\gamma}$, also called \emph{SP-score}, for
scoring the alignment $A$ is
\begin{eqnarray}
\textstyle \costSP{\gamma}[A] = \sum_{h=1}^{k-1} \sum_{i=h+1}^{k}
\costA{\gamma}[A_{\{h, i\}}]\,.\label{exp:def-sp}
\end{eqnarray}

We define $\DistanceSP{\gamma}$ as the optimal function for the
criterion $\costSP{\gamma}$. An alignment $A$ of $S$ such that
$\costSP{\gamma}[A] = \DistanceSP{\gamma}(S)$ is called
\emph{$\costSP{\gamma}$-optimal alignment}.  Regardless its decision
or optimization version, we call the associated problem as \emph{multiple sequence
  alignment problem} ($\MSA$). Formally,

\begin{prob}[Multiple sequence alignment]
  Let $\gamma$ be a fixed scoring matrix. Given a $k$-sequence $S$,
  find $\DistanceSP{\gamma}(S)$.
\end{prob}

In order to compute $\DistanceSP{\gamma}$, we extend the definition of
$\costSP{\gamma}$ considering a column of an alignment $A = [s'_1,
  \ldots, s'_k]$ as its parameter. Thus, $\costSP{\gamma}[A(j)] = \costSP{\gamma}[A[\vec{j}]] =
\sum_{i < h} \pont{\gamma}{s'_{i}(j)}{s'_h (j)}$ assuming that
$\pont{\gamma}{\gap}{\gap} = 0$ and $\vec{j} = [j, \ldots, j]$ and
\begin{align}
  \DistanceSP{\gamma}(S) &= \min_{A \in \mathcal{A}_S} \costSP{\gamma}[A] =
  \left\{
  \begin{array}{ll}
  0 & \mbox{if $S = \emptyset$}\\
  \textstyle \min_{\vec{b} \in \conjuntoBool{k}, \vec{b} \le \vec{n}}
    \Big\{
  \DistanceSP{\gamma}(S(1\!:\!\vec{n} - \vec{b})) +
  \costSP{\gamma}[\vec{b} \cdot S(\vec{n})]\Big\} & \mbox{otherwise}.
  \end{array}
  \right.
  \label{eq:optSP}
\end{align}

Recurrence~(\ref{eq:optSP}) can be computed using a dynamic
programming algorithm, obtaining $D(\vec{\jmath}) =
\DistanceSP{\gamma}(S(1\!:\!\vec{\jmath}))$ for all $\vec{\jmath} \le
\vec{n}$. This task can be performed by generating all indexes of $D$
in lexicographical order, starting with $D(\vec{0}) = 0$, as presented
in Algorithm~\ref{alg-varias1}.

\begin{algorithm}[th!]\caption{}\label{alg-varias1}

  \begin{algorithmic}[1]

    \REQUIRE $S = s_{1}, \ldots, s_{k} \in (\Sigma^{*})^{k}$ 

    \ENSURE $\DistanceSP{\gamma}(S)$

    \vspace{0.2cm}
  
    \STATE $D(\vec{0}) \gets 0$

    \FOR{each $\vec{\jmath} \le \vec{n}$ in lexicographical order} 

      \STATE $D(\vec{\jmath}) \gets \min_{\vec{b} \in \conjuntoBool{k}, \,
        \vec{b} \le \vec{\jmath}} \big\{ D(\vec{\jmath} - \vec{b}) +
      \costSP{\gamma}[\vec{b} \cdot S(\vec{\jmath})] \big\}$

    \ENDFOR
      
    \RETURN $D(\vec{n})$

  \end{algorithmic}
  
\end{algorithm}

Suppose that $\abs{s_i} = n$ for each $i$. Notice that the space to
store the matrix $D$ is $\Theta((n + 1)^{k})$ and thus
Algorithm~\ref{alg-varias1} uses $\Theta((n + 1)^{k})$-space. Besides
that, Algorithm~\ref{alg-varias1} checks, in the worst case,
$\Theta(2^{k})$ entries for computing all entries in the matrix $D$
and each computation spends $\Theta(k^{2})$-time. Therefore, its
running time is $O(2^{k}k^{2}(n + 1)^{k})$. Observe that when the
distance is small, not all entries in $D$ need to
be computed, such as in the Carrillo and Lipman's
algorithm~\cite{carrillo1988}. 

We can also describe an obvious but unusual variant of this problem
that consider a scoring matrix for each
pair of the sequences. Formally, considering an scoring matrix array $\vet{\gamma} = [ 
\gamma^{(12)},
\gamma^{(13)}, \ldots,
\gamma^{(1k)},
\gamma^{(23)}, \ldots, 
\gamma^{(2k)},
\ldots,
\gamma^{((k-1)k)} ]$
and an alignment $A$ of a $k$-sequence,
\begin{eqnarray*}
\costSP{\vec{\gamma}}[A] = \sum_{h=1}^{k-1} \sum_{i=h+1}^{k}
\costA{\gamma^{(hi)}}[A_{\{h, i\}}]\,\label{exp:def-extsp}.
\end{eqnarray*}
Thus, given a $k$-sequence $S$, we ask for finding $\DistanceSP{\vec{\gamma}} (S) = \min_{A \in \mathcal{A}_S} \{ \costSP{\vec{\gamma}}[A] \}$.
Algorithm~\ref{alg-varias1} can be easily modified in order to solve this extended version with same time and space complexity.
This is important here because this version is used as a subroutine of one of the normalized version.

\subsection{$\costSPN{\gamma}{i}$-score for $k$ sequences}

In this section we define a new criteria to normalize the
$\costSP{\gamma}$-score of a multiple alignment.
%% The symbol $\gap$
%% aligned to the same symbol $\gap$ does not contribute to the
%% definition of scoring, and thus this entry is not defined. However,
%% as all the criteria are additive, it is convenient to consider
%% $\pont{\gamma}{\gap}{\gap} = 0$.
The new criteria for aligning
sequences takes into account the length of the alignments according to
the following:
\begin{align}
  \costSPN{\gamma}{1}[A] &= \left\{ \begin{array}{ll}
    0\,, & \mbox{if $\abs{A} = 0$\,,} \\
    \costSP{\gamma}[A]/\abs{A}\,, & \mbox{otherwise}\,, \\
    \end{array} \right. \label{criterion:V1} \\
  \costSPN{\gamma}{2}[A] &= \textstyle \sum_{h=1}^{k-1} \sum_{i=h+1}^{k}
  \costN{\gamma}[A_{\{h, i\}}]\,, \label{criterion:V2} \\
  \costSPN{\gamma}{3}[A] &= \left\{ \begin{array}{ll}
    0\,, & \mbox{if $\abs{A} = 0$\,,} \\
    \costSP{\gamma}[A] \big/ \big(\sum_{h=1}^{k-1} \sum_{i=h+1}^{k}
    |A_{\{h, i\}}| \big)\,, \big. & \mbox{otherwise}\,.
  \end{array} \right. \label{criterion:V3} 
\end{align}

We define $\DistanceNSP{\gamma}{z}$ as the optimal function for the
criterion $\costSPN{\gamma}{z}$, i.e., for a given $k$-sequence $S$
\[
\DistanceNSP{\gamma}{z} (S) = \min_{A \in \mathcal{A}_S} \costSPN{\gamma}{z}[A]. 
\]
An alignment $A$ of $S$ such that
$\costSPN{\gamma}{z}[A] = \DistanceNSP{\gamma}{z}(S)$ is called
\emph{$\costSPN{\gamma}{z}$-optimal alignment}. Moreover, regardless
its decision or optimization version, we establish the \emph{criterion
  $\costSPN{\gamma}{z}$ for the normalized multiple sequence alignment
  problem} ($\NMSA\text{-}z$), for $z = 1, 2, 3$. Formally,
for a fixed scoring matrix $\gamma$ and $z \in \{1, 2, 3\}$,

\begin{prob}[Normalized multiple sequence alignment with score $\costSPN{\gamma}{z}$] \label{pro:NMSA}
  Given a $k$-sequence $S$, find $\DistanceNSP{\gamma}{z}(S)$.
\end{prob}

An interesting question is whether the definitions above represent the same criterion. i.e., would it be possible that the optimal alignment for a given criterion $\costSPN{\gamma}{z}$ also represents an optimal alignment for another criterion $\costSPN{\gamma}{z'}$, $z \not= z'$, regardless of the sequences and scoring matrices? The answer is no and Figure~\ref{fig:fig3} shows examples that support this claim.

\begin{figure}
  \[
  \gamma = 
    \begin{array}{c|cccc}
      & \A & \B & \C & \gap\\
      \hline
      \A & 0 & 9 & 9 & 10\\ 
      \B & 9 & 0 & 9 & 10\\ 
      \C & 9 & 9 & 0 & 10\\ 
      \gap & 10 & 10 & 10 & 0\\ 
      \end{array}
  \hspace{1cm}
    A = \left[
  \begin{array}{c}
      \A \\ \B \\ \C
  \end{array}
    \right]
    \hspace{1cm}
    B = \left[
  \begin{array}{cc}
      \A & \gap \\ \B & \gap \\ \gap & \C
  \end{array}
  \right]
    \hspace{1cm}
    C = \left[
  \begin{array}{ccc}
      \A & \gap & \gap \\ \gap & \B & \gap \\ \gap & \gap & \C
  \end{array}
    \right]
    \]
    \vspace{1cm}
    \[
  \delta = 
    \begin{array}{c|cccc}
      & \A & \B & \C & \gap\\
      \hline
      \A & 0 & 7 & 7 & 9\\ 
      \B & 7 & 0 & 7 & 9\\ 
      \C & 7 & 7 & 0 & 9\\ 
      \gap & 9 & 9 & 9 & 0\\ 
      \end{array}
  \hspace{1cm}
    D = \left[
  \begin{array}{ccc}
    \A & \B & \C \\
    \A & \C & \B \\
    \C & \B & \A
  \end{array}
    \right]
    \hspace{1cm}
    E = \left[
  \begin{array}{cccc}
    \A & \B & \C & \gap\\
    \A & \gap & \C & \B \\
    \C & \B & \A & \gap
  \end{array}
  \right]
    \hspace{1cm}
    F = \left[
  \begin{array}{ccccc}
    \A & \B & \C & \gap & \gap\\
    \A & \gap & \C & \B & \gap\\
    \gap & \gap & \C & \B & \A 
  \end{array}
    \right]
    \]
    \caption{Observe first that $A$, $B$ or $C$ is an $\costSPN{\gamma}{z}$-optimal alignment of 3-sequence $\A, \B, \C$ for each $z$. Besides, since $\costSPN{\gamma}{1}[A] = 27.0, \costSPN{\gamma}{2}[A] = 27.0, \costSPN{\gamma}{3}[A] = 9.0; \costSPN{\gamma}{1}[B] = 24.5, \costSPN{\gamma}{2}[B] = 29.0, \costSPN{\gamma}{3}[B] = 9.8; \costSPN{\gamma}{1}[C] = 20.0, \costSPN{\gamma}{2}[C] = 30.0, \costSPN{\gamma}{3}[C] = 10.0$, we have
that $C$ is an optimal alignment for
      criterion $\costSPN{\gamma}{1}$ but it is not for criteria $\costSPN{\gamma}{2}$ or
      $\costSPN{\gamma}{3}$,
      which implies that an optimal alignment for criterion $\costSPN{\gamma}{1}$
      is different when we compare it to criteria  $\costSPN{\gamma}{2}$ and
      $\costSPN{\gamma}{3}$.
      Now, observe that $D$, $E$ or $F$ is an $\costSPN{\delta}{z}$-optimal alignment of 3-sequence $\A\B\C, \A\C\B, \C\B\A$ for each $z$. Besides, since $\costSPN{\delta}{2}[D] = 16.33, \costSPN{\delta}{3}[D] = 5.44; \costSPN{\delta}{2}[E] = 17.16, \costSPN{\delta}{3}[E] = 5.81; \costSPN{\delta}{2}[F] = 16.20, \costSPN{\delta}{3}[F] = 5.53$, we have that $D$ is an optimal alignment for
      criterion $\costSPN{\delta}{3}$ but it is not for criterion $\costSPN{\delta}{2}$
      which implies that an optimal alignment for criterion $\costSPN{\delta}{2}$
      is different when we compare it to criterion $\costSPN{\delta}{3}$.
    }\label{fig:fig3}
  \end{figure}

\section{Complexity}\label{sec:complexity}

We study now the complexity of the multiple sequence alignment problem
for each new criterion defined in Section~\ref{sec:pre}. We consider
the decision version of the computational problems and we prove
$\NMSA\text{-}z$ is NP-complete for each $z$ even though the following
additional
restrictions for the scoring matrix $\gamma$ hold:
$\pont{\gamma}{\A}{\B} = \pont{\gamma}{\B}{\A}$ and
$\pont{\gamma}{\A}{\B} = 0$ if and only if $\A = \B$ for each pair
$\A, \B \in \Sigma_\gap$. Elias~\cite{Elias2006} shows that, even
considering such restrictions, $\MSA$ is NP-complete.
We start showing a
polynomial time reduction from $\MSA$ to $\NMSA\text{-}z$.

For an instance $(S, C)$ of $\MSA$ ($\NMSA\text{-}z$),
where $S = s_1, \ldots, s_k$ is a $k$-sequence
and $C$ is an integer, $\MSA(S,C)$ ($\NMSA\text{-}z(S,C)$) is the decision problem version asking whether there exists an alignment $A$ 
of $S$ such that $\costSP{\gamma}[A] \le C$ ($\costSPN{\gamma}{i}[A] \le C$).

Consider a fixed alphabet $\Sigma$ and scoring matrix $\gamma$ with
the restrictions above that are $\pont{\gamma}{\A}{\B} = \pont{\gamma}{\B}{\A}$ and
$\pont{\gamma}{\A}{\B} = 0$. We also assume each entry of
$\gamma$ is integer. Let $\sigma \not\in \Sigma_\gap$ be a new
symbol and $\Sigma^\sigma = \Sigma \cup \{\sigma\}$ and $G$
the maximum number in $\gamma$.
We define a scoring matrix $\gamma^\sigma$ such that
$\pont{\gamma^\sigma\!\!\!}{\A}{\B} = \pont{\gamma}{\A}{\B},
\pont{\gamma^\sigma\!\!\!}{\A}{\sigma} =
\pont{\gamma^\sigma\!\!\!}{\sigma}{\A} = G$ and
$\pont{\gamma^\sigma\!\!\!}{\sigma}{\sigma} = 0$ for each pair $\A,
\B \in \Sigma_\gap$.
Notice that $\pont{\gamma^\sigma}{\A}{\B} = \pont{\gamma^\sigma}{\B}{\A}$ and
$\pont{\gamma^\sigma}{\A}{\B} = 0$ if and only if $\A = \B$ for each pair
$\A, \B \in \Sigma^\sigma_\gap$ and therefore $G \ge 1$.
Also, we denote $S^L =
s_1\sigma^L, \ldots, s_k\sigma^L$, where $L = Nk^2MG$, $M =
\max_i\{\abs{s_i}\}$ and $N = \binom{k}{2} M$.

Let $A$ be an alignment of $S^L$.
A $\sigma$-\emph{column} in $A$
is a column where every symbol is equal to $\sigma$.
The \emph{tail} of $A$ is the alignment $A[j+1: \abs{A}]$
if each its column is $\sigma$-columns but $A(j)$ is not;
in this case the \emph{tail length} of $A$ is $\abs{A} - j$ and
the column $j$ is the \emph{tail base}.
We say that an alignment of $S^L$ is
\emph{canonical} if its tail length is~$L$.
If
$A = [s_1'', \ldots, s_k'']$ is the alignment of $S$, then we denote by
$A^L$ the canonical alignment $[s_1''\sigma^L, \ldots, s_k''\sigma^L]$
of $S^L$, i.e., $A^L[1:\abs{A}] = A$ and $A^L[\abs{A}+1:\abs{A^L}] = [\sigma^L, \sigma^L, \ldots, \sigma^L]$.
Notice that $\abs{A^L} \ge L$.
Then, we establish below a lower bound for $\costSP{\gamma^\sigma}[A^L]$ and $\costSP{\gamma}[A]$.

%% Let $B = [(s_1\sigma^L)', \ldots, (s_k\sigma^L)']$ be an alignment of
%% $S^L$. We say that $\abs{\{i : (s_i\sigma^L)'(j) = \sigma\}}$ is the
%% \emph{number (of occurrences) of $\sigma$ in the column $j$} of $B$
%% and that \emph{$j$ is a column $\sigma$} if $k$ is the number of
%% $\sigma$ in the column $j$. If $j$ is not a column $\sigma$ of $B$,
%% but $j+1, \ldots, \abs{B}$ are columns $\sigma$, then we say that
%% $B(j+1 \!:\! \abs{B})$ is the \emph{tail} of the alignment $B$, that
%% the column $j$ is the \emph{tail base} and that $\abs{B} - j$ is the
%% \emph{tail length} of $B$. 

%% Define $l$ as the
%% \emph{tail length} of an alignment $A$ if $A[i,j] = \sigma$ for each
%% $i = 1, \ldots, k$ and $j = \abs{A} - l + 1, \ldots, \abs{A}$, i.e., every
%% symbol in the last $l$ columns of $A$ is $\sigma$.
%% The column $\abs{A} - l$ is the \emph{tail base} of $A$.
%% We say that an
%% alignment of $S^L$ is \emph{canonical} if its tail length is~$L$.

\begin{prop} \label{prop:vSP-upper}
  Let $A$ be an alignment of a $k$-sequence $S = s_1, s_2, \ldots, s_k$. Then,
  $\costSP{\gamma^\sigma}[A^L] = \costSP{\gamma}[A] \le k^2MG$.
\end{prop}

\begin{proof}
  Suppose that $A = [s_1', \ldots, s_k']$ and then $A^L = [s_1'\sigma^L, \ldots, s_k'\sigma^L]$.
  Because $M = \max_i \abs{s_i}$, each alignment that is induced by two
  sequences in $A$ has at most $2M$ columns. Moreover, each entry in
  $\gamma$ is at most $G$. It follows that
  $\costA{\gamma}[s_h', s_i'] = \costA{\gamma}[s_h'\sigma^L,
  s_i'\sigma^L] \le 2MG$ for each pair $h, i$ and then
  \[
  \costSP{\gamma^\sigma}[A^L] = \costSP{\gamma}[A] =
  \sum_{h=1}^{k-1} \sum_{i=h+1}^k \costA{\gamma}[A_{\{h,i\}}] \le
  \sum_{h=1}^{k-1} \sum_{i=h+1}^k 2MG = \textstyle\binom{k}{2}2MG \le
  k^2MG\,.
  \]
\end{proof}

The two following result are useful to prove
Theorem~\ref{theo:NMSAz-NP-complete}, which is the main result of this
section.
  Let $C^1 := C^2 := C/L, C^3 := C/\big(\binom{k}{2}L\big)$ and $L :=
  Nk^2MG$.

\begin{lem} \label{prop:AVSNz-MSA}
  If $C \ge k^2MG$, then $\MSA(S, C) = \NMSA\text{-}z(S^L,
  C^z) = \textsf{Yes}$, for each $z = 1, 2, 3$.
\end{lem}

\begin{proof}
  Suppose that $C \ge k^2MG$.
  Let $A$ be an alignment of $S$.
  From Proposition~\ref{prop:vSP-upper}, $\costSP{\gamma}[A]
  = \costSP{\gamma^\sigma}[A^L] \le k^2MG \le C$ which implies
  that $\MSA(S, C) =
  \textsf{Yes}$. Since $\costSP{\gamma^\sigma}[A^L] \le C$, we have
  that $\costSPN{\gamma^\sigma}{1}[A^L] =
  \costSP{\gamma^\sigma}[A^L]/L \le C/L = C^1$, and then $\NMSA\text{-}1(S^L,
  C^1) = \textsf{Yes}$. Since $\costSP{\gamma^\sigma}[A^L] \le C$
  and $|A_{\{h, i\}}^L| \ge L$, we have
  \[
  \costSPN{\gamma^\sigma}{2}[A^L]
  = \sum_{h = 1}^{k - 1} \sum_{i = k +
    1}^k \frac{\costA{\gamma^\sigma}[A_{\{h, i\}}^L]}{\big|A_{\{h,
      i\}}^L\big|} \le \frac{\displaystyle
    \sum_{h = 1}^{k - 1}
    \sum_{i = h + 1}^k \costA{\gamma^\sigma}[A_{\{h, i\}}^L]}{L} =
  \frac{\costSP{\gamma^\sigma}[A^L]}{L} \le
  \frac{C}{L} = C^2\,,
  \]
  and thus $\NMSA\text{-}2(S^L, C^2) = \textsf{Yes}$. Again, since
  $\costSP{\gamma^\sigma}[A^L] \le C$ and $|A_{\{h, i\}}^L| \ge L$,
  we have
  \[
  \costSPN{\gamma^\sigma}{3}[A^L] =
  \frac{\costSP{\gamma^\sigma}[A^L]}{\displaystyle\sum_{h = 1}^{k -
      1} \sum_{i = k + 1}^k \big|A_{\{h, i\}}^L\big|} \le
  \frac{\costSP{\gamma^\sigma}[A^L]}{\displaystyle\sum_{h = 1}^{k -
      1} \sum_{i = k + 1}^k L} =
  \frac{\costSP{\gamma^\sigma}[A^L]}{\binom{k}{2}L} \le
  \frac{C}{\binom{k}{2}L} = C^3\,,
  \]
  and then $\NMSA\text{-}3(S^L, C^3) = \textsf{Yes}$.

  Therefore, if $C \ge k^2MG$, then $\MSA(S, C) = \NMSA\text{-}z(S^L, C^z) = \textsf{Yes}$ for each $z$.
\end{proof}

\begin{lem} \label{prop:can-opt-align}
  There exists a canonical alignment of $S^L$ which is
  $\costSPN{\gamma^\sigma}{z}$-optimal for each $z$.
\end{lem}
\begin{proof}
  Suppose by contradiction that any canonical alignment of $S^L$ is
  not $\costSPN{\gamma^\sigma}{z}$-optimal. Let $A = [s_1', \ldots,
    s_k']$ be a $\costSPN{\gamma^\sigma}{z}$-optimal alignment of
  $S^L$ with maximum tail length and maximum number of $\sigma$ in the
  tail base. Note that, by hypothesis, $A$ is not canonical.

  Let $q$ be the index of the tail base of $A$. Since $A$ is not
  canonical, the column $q$ contains only symbols $\gap$ and
  $\sigma$. Let $p$ be the greatest index such that $p < q$ and there
  exists an integer $i$ where $s_i'(p) = \sigma$ and $s_i'(q) =
  \gap$. Let $A' = [s_1'', \ldots, s_k'']$ be an alignment of $S^L$
  such that $A'$ is almost the same as $A$, except for columns $p$ and
  $q$, that are defined as following. For each $h$, we have
  \[
  s_h'' = \left\{ \begin{array}{ll}
    s_h'\,, & \text{if $s_h'(p) \neq \sigma$~\text{or}~$s_h'(q) \neq
      \gap$}\,, \\
    s_h'(1 \!:\! p - 1) \cdot \gap \cdot s_h'(p+1 \!:\! q - 1) \cdot
    \sigma \cdot s_h'(q + 1 \!:\! \abs{s_h'})\,, & \text{otherwise}\,.
  \end{array} \right.
  \]

  Observe that either the tail length of $A'$ is greater than the tail
  length of $A$ or the tail lengths of $A$ and $A'$ are the same but
  the number of $\sigma$ in the tail base of $A'$ is greater than this
  number in $A$. Thus, by the choice of $A$, the alignment $A'$ is not
  $\costSPN{\gamma^\sigma}{z}$-optimal.

  Let $h, i$ be integers. We classify the induced alignment $A_{\{h,
    i\}}$ of $A$ as follows:
  \begin{itemize}
  \item \textit{Type 1}: if $\costA{\gamma^\sigma}[A_{\{h, i\}}] =
    \costA{\gamma^\sigma}[A_{\{h, i\}}']$ and $|A_{\{h, i\}}| =
    |A_{\{h, i\}}'|$\,;
  \item \textit{Type 2}: if $\costA{\gamma^\sigma}[A_{\{h, i\}}] \neq
    \costA{\gamma^\sigma}[A_{\{h, i\}}']$ and $|A_{\{h, i\}}| =
    |A_{\{h, i\}}'|$. In this case, the only possibility is that in
    one of the sequences, say $s_ h'$, is such that $s_h'(p) = \sigma$ and
    $s_h'(q) = \gap$, and the other, $s_i'$, is such that $s_i'(p) = x$
    and $s_i'(q) = \sigma$, where $x \in \Sigma$. By hypothesis, $G
    \ge \pont{\gamma}{x}{\gap} = 
    \pont{\gamma}{\gap}{x}$. Then,
    \[
    \costA{\gamma^\sigma}[A_{\{h, i\}}'] =
    \costA{\gamma^\sigma}[A_{\{h, i\}}] - G +
    \pont{\gamma}{x}{\gap} (= \pont{\gamma}{\gap}{x})
    \le
    \costA{\gamma^\sigma}[A_{\{h, i\}}]\,.
    \]
    Therefore, $\costA{\gamma^\sigma}[A_{\{h, i\}}'] \le
    \costA{\gamma^\sigma}[A_{\{h, i\}}]$\,;
  \item \textit{Type 3}: if $|A_{\{h, i\}}| \neq |A_{\{h, i\}}'|$. In
    this case, the only possibility is that in one of the sequences,
    say $s_h'$, is such that $s_h'(p) = \sigma$ and $s_h'(q) = \gap$, and
    the other, $s_i'$, is such that $s_i'(p) = \gap$ and $s_i'(q) =
    \sigma$. It follows that
    \[
    \costA{\gamma^\sigma}[A_{\{h, i\}}'] =
    \costA{\gamma^\sigma}[A_{\{h, i\}}] - 2G \qquad \text{and} \qquad
    |A_{\{h, i\}}| = |A_{\{h, i\}}'| + 1\,.
    \]
  \end{itemize}

  We consider now the case $z = 1$.
  Suppose that $\abs{A'} = \abs{A}$.
  Analyzing the types above,   
  $\costA{\gamma^\sigma}[A_{\{h, i\}}'] \le
  \costA{\gamma^\sigma}[A_{\{h, i\}}]$, which implies that 
  $\costSP{\gamma^\sigma}[A'] \le
  \costSP{\gamma^\sigma}[A]$ and
  $
  \costSPN{\gamma^\sigma}{1}(A') =
  \costSP{\gamma^\sigma}[A']/\abs{A'} \le
  \costSP{\gamma^\sigma}[A]/\abs{A} =
  \costSPN{\gamma^\sigma}{1}[A]$,
  which contradicts $A'$ is not
  $\costSPN{\gamma^\sigma}{1}$-optimal.
  Then, we assume $\abs{A'} \neq \abs{A}$ which implies that
  $\abs{A'} = \abs{A} - 1$ and at least one alignment
  $A_{\{h, i\}}$ is of type 3, meaning that
  $\costA{\gamma^\sigma}[A_{\{h, i\}}'] =
  \costA{\gamma^\sigma}[A_{\{h, i\}}] - 2G$. It follows that
  $\costSP{\gamma^\sigma}[A'] \le \costSP{\gamma^\sigma}[A] -
  2G$. Then,
  \begin{equation} \label{equ:V1}
    \costSPN{\gamma^\sigma}{1}[A'] =
    \frac{\costSP{\gamma^\sigma}[A']}{\abs{A'}} \le
    \frac{\costSP{\gamma^\sigma}[A] - 2G}{\abs{A}-1}\,.
  \end{equation}
  Let $B$ be a canonical alignment. By
  Proposition~\ref{prop:vSP-upper}, we have that
  $\costSP{\gamma^\sigma}[B] \le k^2MG$. By the choice of $A$, we have
  $\costSPN{\gamma^\sigma}{1}[A] \le
  \costSPN{\gamma^\sigma}{1}[B]$. Since $G \ge 1$ and $\abs{B} \ge L =
  Nk^2MG$, then
  \[
  \costSPN{\gamma^\sigma}{1}[A] \le \costSPN{\gamma^\sigma}{1}[B]
  = \frac{\costSP{\gamma^\sigma}[B]}{\abs{B}} \le \frac{k^2MG}{Nk^2MG} =
  \frac{1}{N} \le G\,.
  \]
  Since $\costSP{\gamma^\sigma}[A]/\abs{A} \le G$, we have
  $
  (\costSP{\gamma^\sigma}[A] - G)/(\abs{A}-1) \le
  \costSP{\gamma^\sigma}[A]/\abs{A}$
  which implies, by equation~\eqref{equ:V1}, $G \ge 1$ and by the definition of
  $\costSPN{\gamma^\sigma}{1}[A]$, that
  \[
  \costSPN{\gamma^\sigma}{1}[A'] \le
  \frac{\costSP{\gamma^\sigma}[A] - 2G}{\abs{A}-1} \le
  \frac{\costSP{\gamma^\sigma}[A] - G}{\abs{A}-1} \le
  \frac{\costSP{\gamma^\sigma}[A]}{\abs{A}} =
  \costSPN{\gamma^\sigma}{1}[A]\,,
  \]
  which contradicts again that $A'$ is not
  $\costSPN{\gamma^\sigma}{1}$-optimal. Thus, there exists a
  canonical alignment of $S^L$ which is
  $\costSPN{\gamma^\sigma}{1}$-optimal.

  Now, we consider the case $z = 2$. If an induced alignment
  $A_{\{h, i\}}$ is of type 1 or 2, then
  $\costA{\gamma^\sigma}[A_{\{h, i\}}'] \le
  \costA{\gamma^\sigma}[A_{\{h, i\}}]$ and $|A_{\{h, i\}}'| = |A_{\{h,
    i\}}|$, which implies that
  \begin{equation} \label{equ:vN-1}
    \costN{\gamma^\sigma}[A_{\{h, i\}}'] =
    \frac{\costA{\gamma^\sigma}[A_{\{h, i\}}']}{|A_{\{h, i\}}'|} \le
    \frac{\costA{\gamma^\sigma}[A_{\{h, i\}}]}{|A_{\{h, i\}}|} =
    \costN{\gamma^\sigma}[A_{\{h, i\}}]\,.
  \end{equation}
  If $A_{\{h, i\}}$ is of type 3, then $\costA{\gamma^\sigma}[A_{\{h,
    i\}}'] = \costA{\gamma^\sigma}[A_{\{h, i\}}] - 2G$ and $|A_{\{h,
    i\}}'| = |A_{\{h, i\}}| - 1$ which implies that
  \begin{equation} \label{equ:vN-2}
    \costN{\gamma^\sigma}[A_{\{h, i\}}'] =
    \frac{\costA{\gamma^\sigma}[A_{\{h, i\}}']}{|A_{\{h, i\}}'|} =
    \frac{\costA{\gamma^\sigma}[A_{\{h, i\}}] - 2G}{|A_{\{h, i\}}| -
      1}
    \le
    \frac{\costA{\gamma^\sigma}[A_{\{h, i\}}] - G}{|A_{\{h, i\}}| -  1}
    \le
    \frac{\costA{\gamma^\sigma}[A_{\{h, i\}}]}{|A_{\{h, i\}}|}
    \le \costN{\gamma^\sigma}[A_{\{h, i\}}]\,,
  \end{equation}
  where the first inequality is a consequence of $G \ge 1$ and the second, since $G$ is the maximum value in $\gamma^\sigma$ and
  therefore $G$ is an upper bound to $\costN{\gamma^\sigma}[A_{\{h, i\}}]$, is a consequence of $\costA{\gamma^\sigma}[A_{\{h, i\}}]/|A_{\{h, i\}}| \le G$.
  As a consequence of equations~\eqref{equ:vN-1} and~\eqref{equ:vN-2} we have that $\costSPN{\gamma^\sigma}{2}[A'] \le 
  \costSPN{\gamma^\sigma}{2}[A]$
  contradicting the assumption that $A'$ is not
  $\costSPN{\gamma^\sigma}{2}$-optimal. Thus, there exists a
  canonical alignment of $S^L$ which is
  $\costSPN{\gamma^\sigma}{2}$-optimal.
 
  Finally, we show the case when $z = 3$. We denote $T_j$ the set of
  all pairs $(h, i)$ such that $A_{\{h, i\}}$ is of type $j$. Recall
  that each induced alignment $A_{\{h, i\}}$ of types 1 and 2 are such
  that $\costA{\gamma^\sigma}[A_{\{h, i\}}'] \le
  \costA{\gamma^\sigma}[A_{\{h, i\}}]$ and $|A_{\{h, i\}}'| = |A_{\{h,
    i\}}|$. Thus, the total contribution of the induced alignments of
  types 1 and 2 to the $\costSPN{\gamma^\sigma}{3}$-score is
  \[
  \sum_{\mathclap{(h, i) \in T_1 \cup T_2}}
  \costA{\gamma^\sigma}[A_{\{h, i\}}'] \, \le \quad
  \sum_{\mathclap{(h, i) \in T_1 \cup T_2}}
  \costA{\gamma^\sigma}[A_{\{h, i\}}] \qquad \text{and} \qquad
  \sum_{\mathclap{(h, i) \in T_1 \cup T_2}} |A_{\{h, i\}}'| \, = \quad
  \sum_{\mathclap{(h, i) \in T_1 \cup T_2}} |A_{\{h, i\}}|\,.
  \]
  And since each alignment $A_{\{h, i\}}$ of type 3 is such that
  $\costA{\gamma^\sigma}[A_{\{h, i\}}'] =
  \costA{\gamma^\sigma}[A_{\{h, i\}}] - 2G$ and $|A_{\{h, i\}}'| =
  |A_{\{h, i\}}| - 1$, we have
  \[
  \sum_{\mathclap{(h, i) \in T_3}} \costA{\gamma^\sigma}[A_{\{h,
    i\}}'] \; = \; \sum_{\mathclap{(h, i) \in T_3}} \big(
  \costA{\gamma^\sigma}[A_{\{h, i\}}] - 2G \big) \qquad \text{and}
  \qquad \sum_{\mathclap{(h, i) \in T_3}} |A_{\{h, i\}}'| \; = \;
  \sum_{\mathclap{(h, i) \in T_3}} \big(|A_{\{h, i\}}| - 1\big)\,.
  \]
  It follows that
  \begin{align*} 
    \costSPN{\gamma^\sigma}{3}[A'] &=
    \frac{\costSP{\gamma^\sigma}[A']}{\abs{A'}} \; = \quad
    \frac{\displaystyle\sum_{\mathclap{(h, i) \in T_1 \cup T_2}}
      \costA{\gamma^\sigma}[A_{\{h, i\}}'] + \sum_{\mathclap{(h, i)
          \in T_3}} \costA{\gamma^\sigma}[A_{\{h,
        i\}}']}{\displaystyle\sum_{\mathclap{(h, i) \in T_1 \cup T_2}}
      |A_{\{h, i\}}'| + \sum_{\mathclap{(h, i) \in T_3}} |A_{\{h, i\}}'|}
    \le \;\quad \frac{\displaystyle\sum_{\mathclap{(h, i) \in T_1 \cup
          T_2}} \costA{\gamma^\sigma}[A_{\{h, i\}}] +
      \sum_{\mathclap{(h, i) \in T_3}}
      \big(\costA{\gamma^\sigma}[A_{\{h, i\}}] -
      2G\big)}{\displaystyle\sum_{\mathclap{(h, i) \in T_1 \cup T_2}}
      |A_{\{h, i\}}| + \sum_{\mathclap{(h, i) \in T_3}} \big(|A_{\{h,
        i\}}| - 1\big)} \\
    &=
    \frac{\costSP{\gamma^\sigma}[A] - 2\abs{T_3}G}{\left(\sum_{h = 1}^{k - 1}
    \sum_{i = h + 1}^k \abs{A_{\{h, i\}}}\right) - \abs{T_3}} \le
    \frac{\costSP{\gamma^\sigma}[A] - \abs{T_3}G}{\left(\sum_{h = 1}^{k - 1}
    \sum_{i = h + 1}^k \abs{A_{\{h, i\}}}\right) - \abs{T_3}} \le
    \frac{\costSP{\gamma^\sigma}[A]}{\sum_{h = 1}^{k - 1}
    \sum_{i = h + 1}^k \abs{A_{\{h, i\}}}} =
    \costSPN{\gamma^\sigma}{3}[A]\,,
  \end{align*}
  where the second inequality is a consequence of  $G \ge 1$, 
  and the last inequality is a consequence of
  $\costSP{\gamma^\sigma}[A]/\abs{A} \le G$ since $G$ is
  the maximum value in $\gamma^\sigma$ and then $G$ is an upper bound
  to $\costSPN{\gamma^\sigma}{3}[A]$. Thus,
  $\costSPN{\gamma^\sigma}{3}[A'] \le
  \costSPN{\gamma^\sigma}{3}[A]$, which contradicts the assumption
  that $A'$ is not $\costSPN{\gamma^\sigma}{3}$-optimal. Therefore,
  there exists a canonical alignment of $S^L$ which is
  $\costSPN{\gamma^\sigma}{3}$-optimal.
\end{proof}

\begin{thm} \label{theo:NMSAz-NP-complete}
  $\NMSA\text{-}z$ is NP-complete for each $z$.
\end{thm}
\begin{proof}
  Given a $k$-sequence $S$, an alignment $A$ of $S$
  and a integer $C$, it is easy to check in polynomial time on the
  length of $A$ that $\costSPN{\gamma}{z}[A] \le C$ for $z \in \{1,
  2, 3\}$. Then, $\NMSA\text{-}z$ is in NP.

  Now, we prove that $\MSA(S, C) = \textsf{Yes}$
  if and only if $\NMSA\text{-}z(S^L, C^z) = \textsf{Yes}$ for each $z
  \in \{1, 2, 3\}$.
  If $C \ge k^2MG$, from Lemma~\ref{prop:AVSNz-MSA}
  the theorem is proved.
  Thus, we assume $C < k^2MG$.

  Suppose that $\MSA(S,
  C) = \textsf{Yes}$ and, hence, there exists an alignment $A$ such
  that $\costSP{\gamma}[A] \le C$:
  \begin{align*}
    \!\!\!\!\!\!\costSPN{\gamma^\sigma}{1}[A^L] &\!=\! 
    \frac{\costSP{\gamma^\sigma}[A^L]}{|A^L|} \le
    \frac{\costSP{\gamma^\sigma}[A^L]}{L} =
    \frac{\costSP{\gamma}[A]}{L} \le \frac{C}{L} = C^1\,, \\
    \!\!\!\!\!\!\costSPN{\gamma^\sigma}{2}[A^L] &\!=\! \sum_{h = 1}^{k - 1}
    \sum_{i = k + 1}^k \!\!\! \frac{\costA{\gamma^\sigma}[A_{\{h,
          i\}}^L]}{|A_{\{h, i\}}^L|} \!\le\! \sum_{h = 1}^{k - 1} \sum_{i
      = h + 1}^k \!\!\! \frac{\costA{\gamma^\sigma}[A_{\{h, i\}}^L]}{L}
    \!=\! \frac{\costSP{\gamma^\sigma}[A^L]}{L} \!=\!
    \frac{\costSP{\gamma}[A]}{L} \!\le\! \frac{C}{L} \!=\! C^2,
    \\
    \!\!\!\!\!\!\costSPN{\gamma^\sigma}{3}[A^L] &\!=\! 
    \frac{\costSP{\gamma^\sigma}[A^L]}{\displaystyle \sum_{h = 1}^{k - 1} \sum_{i =
        k + 1}^k |A_{\{h, i\}}^L|} \le
    \frac{\costSP{\gamma^\sigma}[A^L]}{\displaystyle \sum_{h = 1}^{k - 1} \sum_{i =
        k + 1}^k L} =
    \frac{\costSP{\gamma^\sigma}[A^L]}{\binom{k}{2}L} =
    \frac{\costSP{\gamma}[A]}{\binom{k}{2}L} \le
    \frac{C}{\binom{k}{2}L}= C^3\,,
  \end{align*}
  where the first inequality in each equation follows from either
  $A^L$ or each alignment induced by $A^L$ has length at least $L$ and
  the second inequality follows from $\costSP{\gamma}[A] \le
  C$. Thus, if $\MSA(S, C) = \textsf{Yes}$ then
  $\NMSA\text{-}z(S^L, C^z) = \textsf{Yes}$.
  
  Conversely, suppose that $\NMSA\text{-}z(S^L, C^z) = \textsf{Yes}$. It follows
  from Lemma~\ref{prop:can-opt-align} that there exists
  a canonical alignment $A^{L}$ of $S^L$ such that
  $\costSPN{\gamma^\sigma}{z}[A^{L}] \le C^z$ for each $z$. Thus, considering
  $\costSPN{\gamma^\sigma}{1}[A^L] \le C^1$, we have
  \begin{align*}
  \costSP{\gamma}[A] &= \costSP{\gamma^\sigma}[A^L] = (N +
  L)\,\frac{\costSP{\gamma^\sigma}[A^L]}{N + L} \le (N +
  L)\,\frac{\costSP{\gamma^\sigma}[A^L]}{|A^L|} = (N + L)\,
  \costSPN{\gamma^\sigma}{1}[A^L] \\
  &\le (N+L)\, C^1 = (N + L)\,\frac{C}{L} = \frac{NC}{L} + C <
  \frac{Nk^2MG}{L} + C = 1 + C\,,
  \end{align*}
  where the first equality holds since $A^L$ is canonical, the first
  inequality holds since $|A^L| \le N + L$ and the second and the
  third inequalities hold by hypothesis. Considering
  $\costSPN{\gamma^\sigma}{2}[A^L] \le C^2$, we have
  \begin{align*}
    \costSP{\gamma}[A] &= \costSP{\gamma^\sigma}[A^L] = (N + L)\,
    \frac{\costSP{\gamma^\sigma}[A^L]}{N + L} = (N + L) \sum_{h =
      1}^{k - 1} \sum_{i = h + 1}^k
    \frac{\costA{\!\gamma^\sigma}[A_{\{h, i\}}^L]}{N + L} \\
    &\le (N + L) \sum_{h = 1}^{k - 1} \sum_{i = h + 1}^k
    \frac{\costA{\!\gamma^\sigma}[A_{\{h, i\}}^L]}{|A_{\{h, i\}}^L|} = (N +
    L)\,\costSPN{\gamma^\sigma}{2}[A^L] \\
    &\le (N+L)\, C^2 = (N + L)\,\frac{C}{L} = \frac{NC}{L} + C < \frac{Nk^2MG}{L}
    + C = 1 + C\,,
  \end{align*}
  where the first equality holds since $A^L$ is canonical, the first
  inequality holds since, for each $h, i$, $|A_{\{h, i\}}^L| \le N +
  L$, and the second and the third inequalities hold by
  hypothesis. And finally, considering
  $\costSPN{\gamma^\sigma}{3}[A^L] \le C^3$, we have
  \begin{align*}
    \costSP{\gamma}[A] &= \costSP{\gamma^\sigma}[A^L] = \big(N +
    {\textstyle \binom{k}{2}}L\big) \,
    \frac{\costSP{\gamma^\sigma}[A^L]}{N + {\textstyle \binom{k}{2}}L} \\
    &\le \big(N + {\textstyle \binom{k}{2}}L\big)\,
    \frac{\costSP{\gamma^\sigma}[A^L]}{\displaystyle \sum_{h = 1}^{k
        - 1} \sum_{i = h + 1}^k \big|A_{\{h, i\}}^L\big|} = \big(N +
    {\textstyle \binom{k}{2}}L\big)\,\costSPN{\gamma^\sigma}{3}[A^L] \le \big(N + {\textstyle \binom{k}{2}}L\big)\, C^3\\
     &= \big(N + {\textstyle \binom{k}{2}}L\big)\,
    \frac{C}{\binom{k}{2}L} = \frac{NC}{\binom{k}{2}L} + C
    < \frac{N k^2 M G}{\binom{k}{2}L} + C = \frac{1}{\binom{k}{2}} +
    C \le 1 + C\,,
  \end{align*}
  where the first equality holds since $A^L$ is canonical, the first
  inequality holds since the sum of lengths of two sequences induced
  by a canonical alignment is at most $N + \binom{k}{2}L$ and the
  second and the third inequalities hold by hypothesis.
  Therefore, if $\NMSA\text{-}z(S^L, C^z) = \textsf{Yes}$ then
  $\costSP{\gamma}[A] < 1 + C$ for any $z \in \{1, 2, 3\}$. Since
  the entries in the scoring matrix are integers, we have that
  $\costSP{\gamma}[A]$ is an integer. And since $C$ is an integer,
  it follows that $\costSP{\gamma}[A] \le C$.
  
\end{proof}

\section{Exact algorithms} \label{sec:exatos}

In the following sections we describe exact dynamic programming algorithms for $\NMSA\text{-}z$,
with $z = 1, 2, 3$.

\subsection{$\NMSA\text{-}1$} \label{MMSA-1}

Let $S = s_{1}, \ldots, s_{k}$ be a $k$-sequence and $A = [s'_{1},
  \ldots, s'_{k}]$ be an alignment of $S$. As defined in
Equation~(\ref{criterion:V1}), $\costSPN{\gamma}{1}[A]$ takes into
account the length of $A$, and the optimal function is given by
$\DistanceV{\gamma}{1}(S) = \min_{A \in \mathcal{A}_S}
\big\{\costSPN{\gamma}{1}[A] \big\}$.
%% The
%% $\costSPN{\gamma}{1}$-optimal alignment of $S$ is an alignment $A$
%% such that $\costSPN{\gamma}{1}[A] = \DistanceV{\gamma}{1}(S)$.  Thus,
In the optimization version of $\NMSA\text{-}1$,
we are given a $k$-sequence $S$ and we want to compute
$\DistanceV{\gamma}{1}(S)$ for a fixed matrix $\gamma$. We can solve
$\NMSA\text{-}1$ by calculating the minimum SP-score considering every
possible length of an alignment. That is, we compute the entries of a
table $D$ indexed by $\cjtoIndex{\vec{n}} \times \{0, 1, \ldots, N\}$, where
$\vec{n} = [\abs{s_1}, \ldots, \abs{s_k}]$ and $N = \sum_{i = 1}^k |s_i|$. The entry $D(\vet{v}, L)$ stores the score
of an alignment of $S(\vet{v})$ of length $L$ with minimum
SP-score. Notice that $D(\vec{0}, 0) = 0$, $D(\vet{v} \not= \vec{0},
0) = D(\vec{0}, L \not= 0) = \infty$. Therefore, the table entries can
be calculated as:
\[
D(\vet{v}, L)\! = \!
\left\{ \!\! 
\begin{array}{ll}
\!0\,, & \mbox{\!\!\!\!if $\vet{v}\!=\!\vec{0}, L\!=\!0$}\,,\\
\!\infty\,, & \mbox{\!\!\!\!if $\vet{v}\!=\!\vec{0}, L\!\not=\!0$ or $\vet{v}\!\not=\!\vec{0}, L\!=\!0$}\,,\\
\!\min_{\vet{b} \in \mathcal{B}_{k}, \vet{b}\le \vet{v}}
\big\{
D(\vet{v} - \vet{b}, L-1) +  
\costSP{\gamma}[ \vet{b} \cdot S (\vet{v})]
\big\}, & \mbox{\!\!\!\!otherwise}\,.
\end{array}
\right.
\] 

Table $D$ is computed for all possible values of $L = 0 , \ldots,
N$. Consequently, $\DistanceV{\gamma}{1}(S) = \min_{L} \left\{
D(\vet{n}, L)/L\right\}$ is returned. Algorithm~\ref{alg-NMSA-1}
describes this procedure more precisely.

\begin{algorithm}[hbt]\caption{}\label{alg-NMSA-1}
\begin{algorithmic}[1]
\REQUIRE $k$-sequence $S = s_{1}, \ldots, s_{k}$ such that $n_i = \abs{s_i}$ %$\comprimento{s_{i}} = n_{i}$ 
\ENSURE $\DistanceV{\gamma}{1}(S)$ %$\DIST{1}{\matriz}(S)$
\STATE $D(\vec{0},0) \gets 0$
\STATE \textbf{for} each $L \not= 0$ \textbf{do} $D(\vec{0},L) \gets \infty$
\STATE \textbf{for} each $\vet{v} \not= \vec{0}$ \textbf{do} $D(\vet{v}, 0) \gets \infty$
\FOR{each $\vec{0} < \vet{v} \le \vet{n}$ in
  lexicographical order}
  \FOR{each $L \gets 1, 2, \ldots, N$}
    \STATE $D(\vet{v}, L) \gets \min_{\vet{b} \in \conjuntoBool{k}, \vet{b} \le \vet{v}} \big\{ D(\vet{v} - \vet{b}, L-1) + \costSP{\gamma}[\vet{b} \cdot S(\vet{v})] \big\}$ \label{alg-variascrit1-linha5}
  \ENDFOR
\ENDFOR
\RETURN{$\min_{L} \big\{D(\vet{n}, L)/L\big\}$} \label{alg-NMSA-1-linha7}
\end{algorithmic}
\end{algorithm}

Suppose that $n_i = \abs{s_i} = n$ for each $i$ which implies that
$\abs{\cjtoIndex{n}} = (n+1)^k$ and $N = nk$.
Notice that the space
to store the matrix $D$ is $\Theta(N \cdot \abs{\cjtoIndex{n}} = kn \cdot (n+1)^k)$.
The time
consumption of Algorithm~\ref{alg-NMSA-1} corresponds to the time
needed to fill the table $D$ up, plus the running time of
line~\ref{alg-NMSA-1-linha7}. Each entry of $D$ can be computed in
$O((2^k - 1)\cdot {k \choose 2}) = O(2^{k}k^{2})$-time. Therefore, the algorithm spends $O(2^{k}k^{2}
\cdot kn(n+1)^{k}) = O(2^{k}k^{3} (n+1)^{k+1})$-time to compute the entire table $D$.
Line~\ref{alg-NMSA-1-linha7} is computed
in $\Theta(N = kn)$-time. Therefore, the running time of
Algorithm~\ref{alg-NMSA-1} is $O(2^{k}k^{3} (n+1)^{k+1}) +
\Theta(N) = O(2^{k}k^{3} (n+1)^{k+1})$.

\subsection{$\NMSA\text{-}2$} \label{MMSA-2}

Let $S = s_{1}, \ldots, s_{k}$ be a $k$-sequence and $A = [s'_{1},
  \ldots, s'_{k}]$ be an alignment of $S$. As defined in
Equation~(\ref{criterion:V2}), $\costSPN{\gamma}{2}[A]$ takes into
account the lengths of the induced alignments in $A$.
In the optimization version of $\NMSA\text{-}2$, we are given a $k$-sequence $S$ and we want to compute
$\DistanceV{\gamma}{2}(S)$ for a fixed scoring matrix $\gamma$.

Let $\vet{L} = [L_{12}, \ldots, L_{hi}, 
  \ldots, L_{(k-1)k}]$ be a ${k \choose 2}$-vector
indexed by the set of pairs of integers $\{h, i\}$ such that $1 \le h < i \le
k$ and $L_{hi}$ denotes the element of $\vet{L}$ of index $\{h, i\}$.
The lengths of the induced alignments can be
represented by a vector $\vet{L}$. Thus, if $A$ is an alignment and
$|A_{\{h, i\}}| = L_{hi}$ for each pair $h, i$, we say that $\vet{L}$
is the \emph{induced length} of $A$. For a $k$-sequence $S = s_{1},
\ldots, s_{k}$, where $n_i = \abs{s_{i}}$ for each $i$, we define
\[
\mathbb{L} = \big\{ \vet{L} = [ {L_{12}, L_{13}, \ldots, L_{1k},
    L_{23}, \ldots L_{2k}, \ldots, L_{(k-1)k}}] : 0 \le L_{hi} \le
n_{h} + n_{i} \big\}.
\]
Therefore, $\mathbb{L}$ contains the induced length of alignment $A$ of $S$ for
all $A \in \mathcal{A}_S$.
Note that if $n$ is the length of each sequence
in $S$, then $\abs{\mathbb{L}} = (2n + 1)^{k \choose 2}$. Let $\vet{b}
= [ b_{1}, \ldots, b_{k} ]$ be a $k$-vector of bits. Overloading the minus
operator ``$-$'', we define $\vet{L} - \vet{b}$ to be a ${k \choose
  2}$-vector $\vet{L'}$
%which is obtained from $\vet{L}$ and
%from $\vet{b}$
such that
\[
L'_{hi} =
\left\{
\begin{array}{ll}
L_{hi}\,, & \mbox{if $b_{h} = b_{i} = 0$}\,,\\
L_{hi}-1\,, & \mbox{otherwise}\,.
\end{array}
\right.
\]
Observe that if $\vet{L}$ is the induced length of an alignment $A$ of
$S(\vet{v})$ and $\vet{b} $ is a $k$-vector of bits such that $\vet{b}
\cdot S(\vet{v})$ is the last column of $A$, then $\vet{L'} = \vet{L}
- \vet{b}$ is the induced length of the alignment
$A(1\!:\!\abs{A}-1)$.

From a $k \choose 2$-vector $\vec{L}$ and a scoring matrix $\gamma$, we can
define an array of $k \choose 2$ scoring matrices $\vec{\gamma} = \gamma \times \vec{L} =
[\ldots, \gamma^{(hi)}, \ldots ]$ indexed by $\{1, \ldots k\} \times \{1, \ldots, k\}$
such that
\[
\pont{\gamma^{(hi)}}{\A}{\B} = \frac{\pont{\gamma}{\A}{\B}}{L_{hi}}, 
\]
for each $\A, \B \in \Sigma_{\gap}$.
Observe that
\[
\costSPN{\gamma}{2}[A] = \sum_{h=1}^{k-1} \sum_{i=h+1}^{k}
\left( \costN{\gamma}[A_{\{h, i\}}]
(= \frac{\costA{\gamma}[A_{\{h, i\}}]}{\abs{A_{\{h, i\}}}}) \right) =
\costA{\gamma^{(hi)}}[A_{\{h, i\}}] = \costSP{\vec{\gamma}}[A]
\]
where $\vec{\gamma} = \gamma \times [\abs{A_{\{1, 2\}}}, \ldots, \abs{A_{\{k-1, k\}}}]$.
Besides, if $\vec{\mathcal{L}}$ is the induced length 
of a $\costSPN{\gamma}{2}$-optimal alignment of $S$, then
we can compute $\DistanceV{\gamma}{2}$ through the recurrence
\begin{eqnarray*}
D_{\vet{\mathcal{L}}}(\vet{v}, \vet{L}) = 
\left\{
\begin{array}{ll}
0\,, & \mbox{if $\vet{v} = \vec{0}$ and $\vet{L} = \vec{0}$}\,,\\
\infty\,, & \mbox{if $\vet{v} = \vec{0}$ and $\vet{L} \not= \vec{0}$}\,,\\
\infty\,, & \mbox{if $\vet{v} \not= \vec{0}$ and 
$\vet{L} = \vec{0}$}\,,\\
\displaystyle \min_{\vet{b} \in \mathcal{B}_{k}, \vet{b}\le \vet{v}, \vet{b}\le \vet{L}}
\big\{ D_{\vet{\mathcal{L}}}(\vet{v} - \vet{b}, \vet{L} - \vet{b}) +
\costSP{\vet{\gamma}} [\vet{b} \cdot S (\vet{v})] \big\}\,,
& \mbox{otherwise},
\end{array}
\right.
\end{eqnarray*}
where $\vet{b} \le \vet{L}$ is also an overloading, meaning that
$\vet{L} - \vet{b} \ge \vec{0}$, and $\vec{\gamma} = \gamma \times \vec{\mathcal{L}}$.
In this case, $\DistanceV{\gamma}{2}(S) = D_{\vet{\mathcal{L}}}(\vet{v}, \vet{\mathcal{L}})$.

Algorithm~\ref{alg-variascrit2} described below finds $\DistanceV{\gamma}{2}(S)$.  

\begin{algorithm}[hbt]\caption{}\label{alg-variascrit2}
\begin{algorithmic}[1]
\REQUIRE A $k$-sequence $S = s_{1}, \ldots, s_{k}$ such that $n_i = \abs{s_{i}}$ 
\ENSURE $\DistanceV{\gamma}{2}(S)$
\FOR{each $\vet{\mathcal{L}} \in \vet{\mathbb{L}}$}
  \STATE $D_{\vet{\mathcal{L}}}(\vec{0},\vec{0}) \gets 0$
  \STATE \textbf{for} each $\vet{L} \not= \vec{0}$ \textbf{do} $D_{\vet{\mathcal{L}}}(\vec{0},\vet{L}) \gets \infty$
  \STATE \textbf{for} each $\vet{v} \not= \vec{0}$ \textbf{do} $D_{\vet{\mathcal{L}}}(\vet{v},\vec{0}) \gets \infty$
  \STATE $\vet{\gamma} \gets \gamma \times \vec{\mathcal{L}}$
  \FOR{each $\vec{0} < \vet{v} \le \vet{n}$ in lexicographical order}
    \FOR{each $\vet{L} \not= \vec{0}$ in lexicographical order}
      \STATE $D_{\vet{\mathcal{L}}}(\vet{v}, \vet{L}) =
      \min_{\vet{b} \in \conjuntoBool{k}, \vet{b} \le \vet{v}, \vet{b} \le \vet{L}} \big\{
      D_{\vet{\mathcal{L}}}(\vet{v} - \vet{b}, \vet{L} - \vet{b}) +
      \costSP{\vet{\gamma}} [\vet{b} \cdot S(\vet{v})] \big\}$
      \label{alg-variascrit2-linha5}
    \ENDFOR
  \ENDFOR    
\ENDFOR
\RETURN $\min_{\vet{\mathcal{L}} \in \vet{\mathbb{L}}} 
  \{ D_{\vet{\mathcal{L}}} (\vet{n}, \vet{\mathcal{L}}) \}$ \label{alg-variascrit2-linha7}
\end{algorithmic}
\end{algorithm}

For $k$ sequences of length $n$, Algorithm~\ref{alg-variascrit2} needs
$(2n+1)^{k \choose 2}\cdot (n+1)^{k}$ space to store the table
$D_{\vet{\mathcal{L}}}$. For each of the $(2n+1)^{k \choose 2}$ values
$\vet{\mathcal{L}} \in \vet{\mathbb{L}}$, table
$D_{\vet{\mathcal{L}}}$ is recalculated. Since the computation of
each entry takes $O(2^{k}k^{2})$-time, the total time is
\[
O\left( 2^{k}k^{2} \cdot (2n+1)^{k \choose 2} \cdot (2n+1)^{k \choose
  2} (n+1)^{k} = \Big(1 + \frac{1}{2n+1} \Big)^{k} (2n +
1)^{k^{2}}k^{2} \right)\,.
\]
Therefore, if $k \le 2n + 1$, the total time is 
$O\big((2n+1)^{k^{2}}k^{2})\big)$, since
$(1 + 1/(2n + 1))^{k} \le (1 + 1/k)^{k}
\le e \le 2.72$ is constant.

\subsubsection*{Existence of an alignment $A$ for a given $\vet{L}$}

Consider a 3-sequence $S = s_1, s_2, s_3$ with length $\vec{n}=[4, 3, 5]$ and
suppose we are interested in an alignment $A$ of $S$ with induced length equal to $\vec{\mathcal{L}} = [4, 7, 5]$, i.e.,
$\abs{A_{\{1,2\}}} = 4, \abs{A_{\{1,3\}}} = 7, \abs{A_{\{2, 3\}}} = 5$. 
Computing $D_{\vec{\mathcal{L}}}(\vet{v}, \vec{\mathcal{L}})$ as in previous section reveals
that $D_{\vec{\mathcal{L}}}(\vet{v}, \vec{\mathcal{L}}) = \infty$ which means that such alignment doesn't exist. 
In fact, we can check that each symbol in sequence 2 is aligned with one symbol in the first and
one symbol in the third sequence which implies that there are at least 3 symbols of the first sequence
aligned with symbols in the third sequence. Thus, $\abs{A_{\{1,3\}}} \le 6$.
Because, computing $D_{\vec{\mathcal{L}}}(\vet{v}, \vec{\mathcal{L}})$ spends exponential time,
it would be great if we could decide whether there exists an associated alignment with the $\vec{\mathcal{L}}$. 
before computing $D_{\vec{\mathcal{L}}}(\vet{v}, \vec{\mathcal{L}})$. Thus, an interesting problem arises and we define it below.

\medskip

\begin{prob}[Existence of an alignment that is associated with an induced length]  
\label{prob:EAIL}
Given the length $\vet{n}$ of a $k$-sequence and a $k
  \choose 2$-vector $\vet{L}$, decide whether there exists an
  alignment $A$ of $S$ such that $\vet{L}$ is the induced length of
  $A$.
\end{prob}

\medskip

We denote Problem~\ref{prob:EAIL} by {\bf EAIL}.
Considering $\vec{n}$ the length of the $k$ sequence $S = s_1, \ldots, s_k$,
an alternative way to represent $\vec{n}$ is through a matrix of integers $M$ that it is indexed by
$\{1, 2, \ldots, k\}$ and
\[
M(h, i) =
\left\{
\begin{array}{ll}
  \abs{s_h} & \mbox{if $h = i$},\\
  \abs{s_h} + \abs{s_i} - L_{hi} & \mbox{otherwise}
\end{array}
\right.
\]
i.e, $M(h, i)$ is the number of symbols in $s_h$ aligned
with symbols in $s_i$.
Thus, we decide whether exists a collection of
sets $c_{1}, \ldots, c_{k}$ such that
$\abs{c_{h} \cap c_{i}} = M[h, i]$,
where $c_h$ is a set of indices $j$ of some alignment $A$ of $S$ with $A[h, j] \not= \gap$.
%% Here, a substitution in an alignment $A$ of two sequences means a column
%% in $A$ having two symbols in $\Sigma$, i.e., no one in this column is a space.
%% Thus, notice that an alignment
%% of $s, t$ has length $L$ if and only if the amount of substitutions
%% in this alignment is exactly $\abs{s} + \abs{t} - L$. Therefore,
%% Problem~\ref{prob:EAIL} can be easily reformulated by replacing
%% $\vet{L}$ with a $k \choose 2$-vector $\vet{X}$ similar to $\vet{L}$,
%% where $X_{hi} = \abs{s_{h}} + \abs{s_{i}} - L_{hi} = n_{h} + n_{i} -
%% L_{hi}$ for each $h$ and each $i$. Thus, this problem can be rewritten
%% as follows: Given $\vet{n}$ and $\vet{X}$, decide whether there exists
%% an alignment $A$ of some $S = s_{1}, \ldots, s_{k}$, where $n_i =
%% \abs{s_{i}}$ for each $i$, and such that each $A_{\{i,j \}}$ has
%% $X_{ij}$ substitutions.
% 
%% Suppose that we have a matrix of integers $M$ with $k$ rows and $k$
%% columns representing vectors $\vet{n}$ and $\vet{X}$ such that, for
%% each $i$, we have $M[i, i] = n_{i}$ and, for each pair $h, i$, we have
%% $M[h, i] = X_{h, i}$. Hence, we can reformulated
%% Problem~\ref{prob:EAIL} as follows: Given an integer matrix $M$ with
%% $k$ rows and $k$ columns, decide whether there exists a collection of
%% sets $c_{1}, \ldots, c_{k}$ such that $\abs{c_{h} \cap c_{i}} = M[h,
%%   i]$.
%% Each $c_i$ represents a column in an alignment of $S$.
This different way to see the problem is exactly another that is known as
\emph{Recognizing Intersection Patterns} ({\bf RIP}).
Notice that the instance of {\bf RIP} can be obtained in linear time
and no extra space from the instance of {\bf EAIL}.

Thus, consider an example where
$\vet{n} = [5, 5, 5]$ and $L_{12}= L_{13} = 8$ and $L_{23} =
10$. In this case, {\bf EAIL} returns \textsf{Yes} since
the induced alignments by the following
alignment with $11$ columns
\[
\left[
\begin{array}{ccccccccccc}
s_{1} (1) & s_{1} (2) & s_{1} (3) & s_{1} (4) & s_{1} (5) & 
\gap & \gap & \gap & \gap & \gap & \gap\\ 
s_{2} (1) & s_{2} (2) & \gap & \gap & \gap
& s_{2} (3) & s_{2} (4) & s_{2} (5) & \gap & \gap & \gap \\ 
\gap & \gap & s_{3} (1) & s_{3} (2) & \gap & \gap & \gap 
& \gap & s_{3} (3) & s_{3} (4) & s_{3} (5) 
\end{array}
\right]
\]
respects required restrictions.
On the other hand, in the alternative {\bf RIP} formulation of
this instance, we want to find a collection of three sets for the
matrix
\[
M =
\begin{array}{c|ccc}
& 1 & 2 & 3\\
\hline
1 & 5 & 2 & 2\\
2 & 2 & 5 & 0\\
3 & 2 & 0 & 5
\end{array}
\]
that satisfies the aforementioned property. The answer to {\bf RIP} is
\textsf{Yes}, as we can see in Figure~\ref{RIPyes}.
\begin{figure}[ht]
\begin{center}
\scalebox{.6}{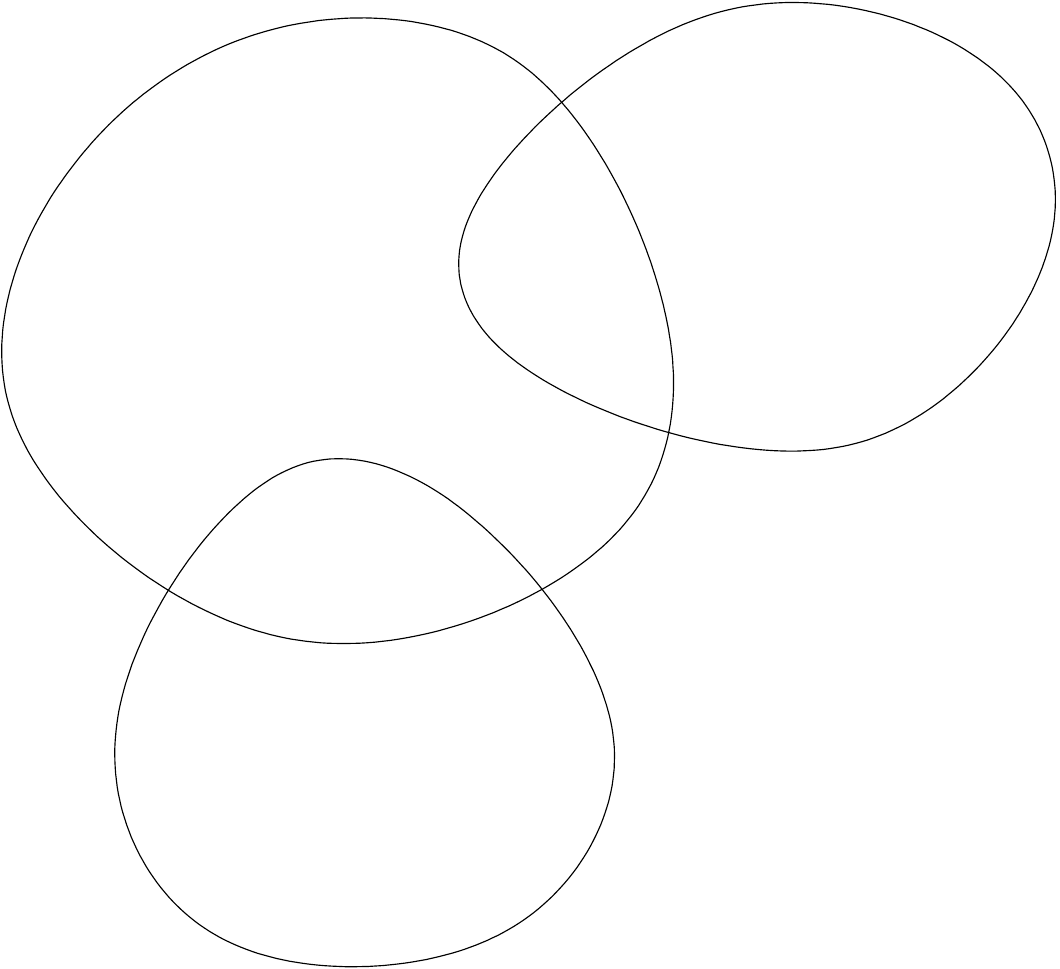}
\end{center}
\caption{An instance of {\bf RIP} for which the answer is \textsf{Yes}.}
\label{RIPyes}
\end{figure}

In another example, suppose that $\vet{n} = [5, 5, 5]$, $L_{12}=
L_{13} = 7$ and $L_{23} = 10$. In this case the answer is
\textsf{No}. To check it, suppose by contradiction that there is an
alignment $A$ of a $3$-tuple $s_{1}, s_{2}, s_{3}$ for this
instance.
Since $L_{23} = 10$, we have that $s_2$ and $s_3$ has no ($5+5-10=0$) aligned symbol.
Since $L_{12} = 7$, then $s_{1}$ must have $5 + 5 - 7 = 3$ symbols 
aligned with $s_2$ that cannot be aligned with any symbol of $s_{3}$
because $s_2$ and $s_3$ has no symbol aligned.
On the other hand, since $L_{13} = 7$, we have that $s_{1}$ and $s_{3}$ must
have $3$ aligned symbols.
Since $s_1$ has 3 symbols not alignment and 3 symbols aligned with $s_3$,
it follows that $s_1$ has at least $6$ symbols, which is a
contradiction.

Chv\'atal~\cite{chvatal1980} showed that for a special class of
matrices $M$ where $M[i, i] = 3$ for every $i$, {\bf RIP} is
NP-complete. Therefore, we have the following result
\begin{thm}
  {\bf EAIL} is NP-complete when $\abs{s} = 3$ for each sequence in $k$-sequence
  and it is NP-hard if its lengths are arbitrary.
  \end{thm}

\subsection{$\NMSA\text{-}3$}

Let $S = s_{1}, \ldots, s_{k}$ be a $k$-sequence and $A = [s'_{1},
  \ldots, s'_{k}]$ be an alignment of $S$. As defined in
Equation~(\ref{criterion:V3}), $\costSPN{\gamma}{3}[A]$ takes into
account the lengths of induced alignments of $A$.
In the optimization version of $\NMSA\text{-}3$,
given a $k$-tuple $S$, the task is to determine $\DistanceV{\gamma}{3}(S)$ for a fixed matrix $\gamma$.

Let $N = \sum_i \abs{s_i}$. Notice that an alignment $A$ that has exactly only one symbol 
different of space in each column is such that $\sum_{h = 1}^{k-1} \sum_{i=h + 1}^{k} \abs{A_{(h, i)}} = (k-1)N$.

Here, each entry $D(\vet{v}, L)$ of $D$ stores the SP-score of an
alignment $A$ of the prefix $S(\vet{v})$ with the minimum SP-score
such that $\sum_{i < h} |A_{\{i,h\}}| = L$.
The Boolean vectors $\vet{b}$ are used to represent the contribution
to the sum of the lengths of the induced alignments. Thus, we define
$\|\vet{b}\| = {k \choose 2} - \sum_{h < i, b_{h} = b_{i} = 0}
1$. Notice that if $\vet{b} \cdot S (\vet{v})$ is the last column of an
alignment $A$ and $L = \sum_{h=1}^{k-1} \sum_{i= h+ 1}^{k} |A_{\{h,
  i\}}|$ is the sum of the lengths of the alignments induced by $A$,
then the sum of the lengths of the alignments induced by
$A(1\!:\!\abs{A}-1)$ is $L - \|\vet{b}\|$. Therefore,
\[
D(\vet{v}, L) = 
\left\{
\begin{array}{ll}
0\,, & \mbox{if $\vet{v} = \vec{0}$, $L = 0$}\,,\\
\infty\,, & \mbox{if $\vet{v} = \vec{0}$, $L \not= 0$ or $\vet{v} \not= \vec{0}$, $L = 0$}\,,\\
\displaystyle \min_{\vet{b} \in \mathcal{B}_{k}, \vet{b}\le \vet{v}, \| \vet{b} \| \le L}
\left\{
D(\vet{v} - \vet{b}, L - \| \vet{b} \|) +  
\costSP{\gamma}[\vet{b} \cdot S (\vet{v})] \right\}\,, & \mbox{otherwise}\,.
\end{array}
\right.
\]

Algorithm~\ref{alg-variascrit3} provides more details about the
procedure for computing $\DistanceV{\gamma}{3}$.

\begin{algorithm}[H]\caption{}\label{alg-variascrit3}
\begin{algorithmic}[1]
\REQUIRE a $k$-sequence $S = s_{1}, \ldots, s_{k}$ such that $n_i =
\abs{s_{i}}$
\ENSURE $\DistanceV{\gamma}{3}(S)$

\STATE $D(\vec{0},0) \gets 0$
\STATE \textbf{for} each $L \not= 0$ \textbf{do} $D(\vec{0},L) \gets \infty$
\STATE \textbf{for} each $\vet{v} \not= \vec{0}$ \textbf{do} $D(\vet{v},0) \gets \infty$
\FOR{each $\vec{0} < \vet{v} \le \vet{n}$ in lexicographical order}
  \FOR{$L \gets 1, 2, \ldots, N(k-1)$}
    \STATE $D(\vet{v}, L) \gets \min_{\vet{b} \in
      \conjuntoBool{k}, \vet{b} \le \vet{v}, \|\vet{b}\| \le L}
    \big\{
    D(\vet{v} - \vet{b}, L-\| \vet{b} \|) + 
    \costSP{\gamma} [ \vet{b} \cdot S(\vet{v}) ]
    \big\}$ \label{alg-variascrit3-linha5}
  \ENDFOR
\ENDFOR
\RETURN $\min_{L}\big\{ D(\vet{n}, L)/L \big\}$ \label{alg-variascrit3-linha6}
\end{algorithmic}
\end{algorithm}

Assume that all sequences in $S$ have length $n$.
In this case,
%% Considering each $L = 0, 1, \ldots, {k \choose 2} (2n)
%% (= nk^{2} - nk)$, we compute 
%% \[
%% \DistanceSP{\gamma}(S) = \min_{L} \left\{ \frac{D(\vet{n}, L)}{L}
%% \right\}\,.
%% \]
%% Thus,
table $D$ has $(nk^{2} - nk + 1) \cdot
(n+1)^{k} = O(k^{2} (n+1)^{ k+1})$ entries. Since the time required to
determine each entry of $D$ is $O(2^{k}k^{2})$, the running time of
Algorithm~\ref{alg-variascrit3} is $O(2^{k}k^{2} \cdot k^{2} (n+1)^{ k+1}) = O(2^{k} k^{4} (n+1)^{k+1})$.

\section{Approximation algorithms for $\MSA$ and $\NMSA\text{-}2$}\label{sec:approx}

Gusfield~\cite{gusfield1993} described a 2-approximation algorithm for $\MSA$
assuming that $\gamma \in \metricC$. In this section, we adapt
Gusfield's algorithm, thus proposing a 6-approximation algorithm for $\MSA$
when $\gamma \in \metricA$ and a 12-approximation algorithm for
$\NMSA\text{-}2$ problem when $\gamma \in \metricN$.

We consider here a scoring function $\generalcost = \costA{\gamma}$
and $\generalcost = \costN{\gamma}$ when $\gamma \in \metricA$
and $\gamma \in \metricN$, respectively.
Also, for a 2-sequence $s, t$, define
$\opt (s, t) = \min_{A \in \mathcal{A}_{s, t}} \{ \generalcost[A]\}$ and an
$\generalcost$-\emph{optimal alignment} of $s, t$ is an alignment $A$ such 
$\generalcost[A] = \opt(s, t)$.
It follows from~\cite{araujo2006} that $\opt$
% satisfies
% reflexivity and strict positiveness, symmetry and triangle inequality
% properties.
is a metric on $\Sigma^*$.
For a $k$-sequence $S$, define
$\Generalcost[A] = \sum_{h=1}^{k-1} \sum_{i= h+1}^k
\generalcost[A_{\{ h, i \}}]$ and $\Opt(S) = \min_{A \in
  \mathcal{A}_S} \generalcost(A)$ and 
a
\emph{$\Generalcost$-optimal alignment} of $S$ is an alignment $A$ such that
$\Generalcost[A] = \Opt(S)$.

\newcommand{\mystar}{\protect\ensuremath{X}}

Let $c$ be an integer, $1 \le c \le k$. A \emph{star $\mystar$ with
  center $c$} (also called \emph{$c$-star}) of the $k$-sequence
$S = s_1, \ldots, s_k$
is a collection of $k-1$ alignments: $\mystar_{h} = [s'_h,
  s^h_c]$ of $s_{h}, s_{c}$ for each $h < c$ and 
$\mystar_{h} = [s^h_c, s'_h]$ of $s_{c}, s_{h}$ for each $h > c$. The set
of all stars with center $c$ is denoted by $\mathcal{X}_c$. The score of the $c$-star
$\mystar$ is $\cost(\mystar) = \sum_{h \not= c}
\generalcost[\mystar_{h}]$ and an \emph{$\generalcost$-optimal star}
is one whose score is $\optStar(S) = \min_{\mystar \in \mathcal{X}_c,
  c \in \mathbb{N}} \{\cost(\mystar)\}$.
Notice that in a $\generalcost$-optimal $c$-star,
$\generalcost(\mystar_h) = \opt (s_h, s_c)$ if $h < c$ and $\generalcost(\mystar_h) = \opt (s_c, s_h)$ if $c < h$ and because the symmetry property of $\opt$, we have
\[
\optStar(S) = \min_c \left\{ \sum_{h \not= c} \opt (s_h, s_c)
\right\},
\]
and if $\abs{s} \le n$ for each $s \in S$,
$\optStar(S)$ can be computed in $O(k^2 n^2)$-time
when $\generalcost = \costA{\gamma}$ and
$O(k^2 n^3)$-time when $\generalcost = \costN{\gamma}$.

We say that alignment $A$ of $S$ and $c$-star $\mystar$ are \emph{compatible}
($A$ is compatible with $\mystar$ and $\mystar$ is compatible with $A$)
in $S$ when either $A_{\{ h, c \}}$ (when $h < c$) or $A_{\{ c, h \}}$
(when $c < h$) is equal to
$\mystar_h$ for each $h$.
Given the alignment $A$ of $S$, it is easy to obtain the $c$-star
$\mystar$ compatible with $A$ (and there exists only one) considering fixed
$c$.
On the other hand, an
important known result in alignment studies from Feng and
Doolitte~\cite{feng1987progressive} is that we can find an alignment
$A$ that is compatible with a given $c$-star $\mystar$ in $O(kn)$,
where $n \le \abs{s}$ for each sequence $s$ in $S$. In general in this case,
there exists many compatible alignments with $\mystar$.

It is easy to adapt the following result from
Gusfield~\cite{gusfield1993} to a $k$-sequence.
\begin{lem}\label{lemma:approx:2}
Given a $k$-sequence $S$, 
\[
\optStar(S) \le \frac{2}{k} \cdot \Opt(S)\,.
\]
\end{lem}

\begin{proof}
  Let $\mystar$ be a $\generalcost$-optimal star of $S$ and
  $c$ its center, and $A$ an
  $\generalcost$-optimal alignment of $S$. Then,
  \begin{align}
    k \cdot \optStar(S) &= 
    k \cdot \cost(\mystar) = 
    \sum_{h = 1}^{k} \cost(\mystar)
    = \sum_{h = 1}^{k} \sum_{h \not= c} \generalcost[\mystar_{h}]
    = \sum_{h = 1}^{k} \sum_{h \not= c} \opt (s_h, s_c) \label{eq:star1}\\
    & \le 2 \cdot \sum_{h=1}^{k-1} \sum_{i=h+1}^k \opt (s_h, s_i)\label{eq:star111}\\
    &\le 2 \cdot \sum_{i= 1}^{k} \sum_{h \not= i} \generalcost[A_{\{ h, i \}}] = 2 \, \Generalcost[A] =  2 \cdot \Opt(S)\,, \label{eq:star1111}
  \end{align}
  where~(\ref{eq:star1}) follows from the definition of a
  star and from the optimality of star $\mystar$;
  (\ref{eq:star111}) follows from triangle inequality of $\opt$;
  and (\ref{eq:star1111}) follows from the optimality of alignment $A$.
  %% and~(\ref{eq:star3}) follows from~(\ref{exp:approxinitial}).  It
  
  Therefore, $\optStar \le (2/k) \cdot \opt(S)$.
\end{proof}

Let $A = [s', t']$ 
be an alignment of a 2-sequence $s, t$.  We say that a column $j$ is \emph{splittable
  in $A$} if $s'(j) \not= \gap$, $t'(j) \not= \gap$ and $\min\{
\pont{\gamma}{t'(j)}{\gap}, \pont{\gamma}{s'(j)}{\gap}\} \le
\pont{\gamma}{s'(j)}{t'(j)}$. Let $J := \{ j_i \in \mathbb{N}: 1 \le
j_{1} < \cdots < j_{m} \le \abs{A} \mbox{and $j_i$ is splittable in
  $A$}\}$. An \emph{$A$-splitting} is the alignment
\[
\left[\begin{array}{cccccccccc}
s'(1\!:\!j_{1}-1) & s'(j_{1}) & \gap &
s'(j_{1}+1\!:\!j_{2}-1) & s'(j_{2}) & \gap &
\ldots & s'(j_m) & - & s'(j_{m} + 1\!:\!\abs{A})\\
t'(1\!:\!j_{1}-1) & \gap & t'(j_{1}) &
t'(j_{1}+1\!:\!j_{2}-1) & \gap & t'(j_{2}) &
\ldots & - & t'(j_{m}) & t'(j_{m} + 1\!:\!\abs{A})
\end{array}\right].
\]
We say that $J$ is \emph{required to split $A$}. The following
proposition is used to check properties of an $A$-splitting.

\begin{prop}\label{fact:approx:0}
Consider $\gamma \in \metricA$ and $\A, \B \in \Sigma$. If
$\pont{\gamma}{\A}{\gap} > \pont{\gamma}{\A}{\B}$ or
$\pont{\gamma}{\A}{\gap} > \pont{\gamma}{\B}{\A}$, then
$\pont{\gamma}{\A}{\B} = \pont{\gamma}{\B}{\A}$.
\end{prop}

\begin{proof}
Since $\gamma \in \metricA$, we have that $\pont{\gamma}{\A}{\gap} =
\pont{\gamma}{\gap}{\A} > 0$ and $\pont{\gamma}{\B}{\gap} =
\pont{\gamma}{\gap}{\B} >0$. Suppose that $\pont{\gamma}{\A}{\gap} >
\pont{\gamma}{\A}{\B}$. Then,
$\pont{\gamma}{\A}{\gap}+\pont{\gamma}{\gap}{\B} >
\pont{\gamma}{\A}{\gap} > \pont{\gamma}{\A}{\B}$, and we have that
$\pont{\gamma}{\A}{\B} = \pont{\gamma}{\B}{\A}$ since $\gamma \in
\metricA$. Assume now that $\pont{\gamma}{\A}{\gap} >
\pont{\gamma}{\B}{\A}$. It follows that $\pont{\gamma}{\B}{\gap} +
\pont{\gamma}{\gap}{\A} > \pont{\gamma}{\gap}{\A} =
\pont{\gamma}{\A}{\gap} > \pont{\gamma}{\B}{\A}$, which implies that
$\pont{\gamma}{\A}{\B} = \pont{\gamma}{\B}{\A}$ since $\gamma \in
\metricA$.
\end{proof}

\newcommand{\secstar}{\protect\ensuremath{Y}}

Let $\mystar = \{\mystar_1, \ldots, \mystar_{c-1}, \mystar_{c+1},
\mystar_k\}$ be a $c$-star. A $\mystar$-\emph{starsplitting} is the
$c$-star $\secstar = \{ \secstar_1, \ldots, \secstar_{c-1},
\secstar_{c+1}, \secstar_k \}$ where $\secstar_j$ is the
$\mystar_j$-splitting for each $j$. The next result shows that the
$\generalcost$-score of the star $\secstar$ is bounded by the
$\generalcost$-score of star $\mystar$ when $\gamma \in \metricA$ and
$\generalcost = \costA{\gamma}$ or $\gamma \in \metricN$ and
$\generalcost = \costN{\gamma}$.

\begin{lem}\label{proposition:approx:1}
  Let $S$ be a $k$-sequence, $\mystar$ be a
  star of $S$, $\secstar$ be the $\mystar$-starsplitting and
  $\generalcost$ be a function to score alignments. Consider $\gamma
  \in \metricA$ and $\generalcost = \costA{\gamma}$ or $\generalcost =
  \costN{\gamma}$. Then, $\secstar$ is also a $c$-star and
  \[
  \cost(\secstar) \le 3 \cdot \cost(\mystar)\,.
  \]
\end{lem}

\begin{proof}
  Consider an alignment $\mystar_{h} = [s', t'] \in \mystar$ and a
  set $J$ which is 
  required to split $\mystar_h$. Then, 
  \begin{align}
  \costA{\gamma}[\secstar_h] 
  &= 
  \costA{\gamma}[\mystar_h]
  + \sum_{j \in J}
  \left( \pont{\gamma}{s'(j)}{\gap} + \pont{\gamma}{\gap}{t'(j)} -
  \pont{\gamma}{s'(j)}{t'(j)} \right) \nonumber \\
  %\label{proposition:approx:1:exp:1} \\
  &\le \costA{\gamma}[\mystar_h]
  + \sum_{j \in J}
  \left(
  \pont{\gamma}{s'(j)}{\gap} +
  \pont{\gamma}{\gap}{t'(j)}
  \right) \le 
  \costA{\gamma}[\mystar_h]
  + 2 \cdot \sum_{j \in J} \min \{
    \pont{\gamma}{s'(j)}{\gap},
  \pont{\gamma}{\gap}{t'(j)}
  \}
  \label{proposition:approx:1:exp:2} \\
  &
  \le 
  \costA{\gamma}[\mystar_h]
  + 2 \cdot \sum_{j \in J} \pont{\gamma}{s'(j)}{t'(j)},
  \label{proposition:approx:1:exp:5} \le
  \costA{\gamma}[\mystar_h]
  + 2 \cdot \costA{\gamma}[\mystar_h] =
  3 \cdot \costA{\gamma}[\mystar_h]\,,
  % \label{proposition:approx:1:exp:55}
  \end{align}
  where~(\ref{proposition:approx:1:exp:2})
  hold because, since $\gamma \in \metricA$,
  $\pont{\gamma}{s'(j)}{t'(j)}$ and 
  (\ref{proposition:approx:1:exp:5}) hold because $j$ is splittable.
  Thus, $\costA{\gamma}[\secstar_h] \le 3 \cdot \costA{\gamma}[\mystar_h]$.
  Furthermore,
  \begin{align*}
    \costN{\gamma}[\secstar_h] &= \frac{\costA{\gamma}[\secstar_h]}
          {\abs{\secstar_h}} \le \frac{3 \cdot
            \costA{\gamma}[\mystar_h]} {\abs{\secstar_h}}
          %% = \frac{3
          %%   \cdot\costA{\gamma}[\mystar_h]} {\abs{\mystar_h} +
          %%   \abs{J}}
          \le \frac{3 \cdot\costA{\gamma}[\mystar_h]}
          {\abs{\mystar_h}} = 3 \cdot \costN{\gamma}[\mystar_h]\,.
  \end{align*}
  Hence,
  $\generalcost(\secstar_{h}) \le 3 \cdot
  \generalcost[\mystar_h]$ when
  .$\generalcost = \costA{\gamma}$
  or $\generalcost = \costA{\gamma}$ and $\gamma \in \metricA$
  which implies that
  \[
  \cost (\secstar) = \sum_{h \not= c} \generalcost [\secstar_{h}]
  \le
  3 \cdot \sum_{h \not= c} \generalcost [\mystar_{h}] =
  3 \cdot \cost (\mystar)\,.
  \]
\end{proof}

Notice that the time consumption for computing an $\mystar$-splitting
from $\mystar$ is $O(kn)$ when $\abs{s} \le n$ for each $s \in S$.

Considering a star $\mystar$ with center $c$ of
$S = s_1, \ldots, s_k$, there can
exist many compatible alignments with a star $\secstar$
which is a $\mystar$-splitting. Let \textsc{CompatibleAlign} be a
subroutine that receives the star $\secstar$ and returns an
alignment $A$ compatible with $\secstar$. It is quite simple: if
symbols $s_h (j_1)$ and $s_c(j_2)$ are aligned in $\mystar_h$, they
are also aligned in $A$. Otherwise, $s_h (j)$ aligns only with $\gap$
in $A$. This property is enough to guarantee the approximation factor
of $\MSA$ and $\NMSA\text{-}2$.
As an example, for $S = \A\A\A, \B\B\B\B\B, \C\C, \D\D\D, \E\E\E\E\E\E$
and
\[
\mystar=
\left\{
\left[
\begin{array}{cccc}
\A & \A & \A & \gap\\
\gap & \D & \D & \D \\
\end{array}
\right],
\left[
\begin{array}{cccccc}
\B & \B & \gap & \B & \B & \B\\
\gap & \D & \D & \gap & \gap & \D
\end{array}
\right],
\left[
\begin{array}{ccc}
\C & \C & \gap\\
\D & \D & \D
\end{array}
\right],
\left[
\begin{array}{ccccccc}
\gap & \gap & \gap & \D & \D & \D& \gap \\
\E & \E &\E & \E & \gap &\E & \E
\end{array}
\right]
\right\}
\]
a star with center $4$, we obtain the alignment
\[
  \left[
    \begin{array}{ccccccccccc}
      \A &
      \gap &
      \gap &
      \gap &
      \gap &
      \A &
      \A &
      \gap &
      \gap &
      \gap &
      \gap \\
      \gap &
      \B &
      \gap &
      \gap &
      \gap &
      \B &
      \gap &
      \B &
      \B &
      \B &
      \gap\\
      \gap &
      \gap &
      \gap &
      \gap &
      \gap &
      \C &
      \C &
      \gap &
      \gap &
      \gap &
      \gap\\
      \gap &
      \gap &
      \gap &
      \gap &
      \gap &
      \D &
      \D &
      \gap &
      \gap &
      \D &
      \gap\\
      \gap &
      \gap &
      \E &
      \E &
      \E &
      \E &
      \gap &
      \gap &
      \gap &
      \E &
      \E
\end{array}
\right]\,.
\]

Let $\qmax := \max_{\A \in \Sigma}\{\pont{\gamma}{\A}{\gap},
\pont{\gamma}{\gap}{\A} \}$ and consider the following result.

\begin{prop}\label{proposition:approx:2}
  Let $S$ be a $k$-sequence, $\mystar$ be a star of $S$ with center $c$ and
  $\secstar$ be the $\mystar$-starsplitting. Assume that $\gamma \in
  \metricA$ and that $\textsc{Compati\-bleAlign}(\secstar)$ returns $A
  = [s'_{1}, \ldots, s'_{k}]$. If $h \not= c$ and $i \not= c$, we
  have that
  \begin{enumerate}
  \item[(i)] $\pont{\gamma}{s'_{h}(j)}{s'_{i}(j)} \le
    \pont{\gamma}{s'_ {h}(j)}{s'_{c}(j)} +
    \pont{\gamma}{s'_{c}(j)}{s'_{i}(j)}$ for each $j = 1, \ldots,
    \abs{A}$, and
  \item[(ii)] $\costN{\gamma}[A_{\{h, i\}}] \le 2\cdot\qmax$.
  \end{enumerate}
\end{prop}

\begin{proof}
  Assume that $s'_h(j) = \A, s'_i(j) = \B$ and $s'_c(j) = \C$.
  
  First we show that (\textit{i}) $\pont{\gamma}{\A}{\B} \le
  \pont{\gamma}{\A}{\C} + \pont{\gamma}{\C}{\B}$ for each $j = 1,
  \ldots, \abs{A}$, by analyzing all possible values of $\A, \B$ and
  $\C$. The case when $\A = \gap$ or $\B = \gap$ can be checked by
  definition of $\gamma \in \metricA$. Thus, we assume that $\A \not=
  \gap$ and $\B \not= \gap$, which implies by the alignment
  construction in \textsc{CompatibleAlign}, that $\C \not= \gap$.
  Since $\A \not= \gap$, $\B \not= \gap$, $\C \not= \gap$ and
  $\secstar$ is a starsplitting, we have that
  $\pont{\gamma}{\A}{\gap} > \pont{\gamma}{\C}{\A}$,
  $\pont{\gamma}{\B}{ \gap} > \pont{\gamma}{\C}{\B}$ and, since
  $\gamma \in \metricA$, $\pont{\gamma}{\gap}{\B} =
  \pont{\gamma}{\B}{\gap} > \pont{\gamma}{\C}{\B}$. Since
  $\pont{\gamma}{\A}{\gap} > \pont{\gamma}{\C}{\A}$, it follows from
  Proposition~\ref{fact:approx:0} that $\pont{\gamma}{\A}{\gap} >
  \pont{\gamma}{\A}{\C}$. Hence, $\pont{\gamma}{\A}{\gap} +
  \pont{\gamma}{\gap}{\B} > \pont{\gamma}{\A}{\C} +
  \pont{\gamma}{\C}{\B}$, which implies from the definition of
  $\metricA$ that $\pont{\gamma}{\A}{\B} \le \pont{\gamma}{\A}{\C}
  +\pont{\gamma}{\C}{\B}$.

  Finally, we show~(\textit{ii}). Here, it is enough to prove that
  $\pont{\gamma}{\A}{\B} \le 2 \cdot \qmax$ for each column $[\A, \B]$
  of $\costN{\gamma}[A_{\{h, i\}}]$. Again, the case when $\A = \gap$
  or $\B = \gap$ can easily be checked. Thus, assume that $\A \not=
  \gap$ and $\B \not= \gap$ which implies by construction that $\C
  \not= \gap$. Since $\secstar$ is a splitting, it follows that
  $\pont{\gamma }{\C}{\A} < \pont{\gamma}{\A}{\gap}$ and
  $\pont{\gamma}{\C}{\B} < \pont{\gamma}{\B}{\gap}$. Since
  $\pont{\gamma}{\C}{\A} < \pont{\gamma}{\A}{\gap}$, it follows from
  Preposition~\ref{fact:approx:0} that $\pont{\gamma}{\A}{\C} <
  \pont{\gamma}{\A}{\gap}$. It follows from~(\textit{i}) that
  $\pont{\gamma}{\A}{\B} \le \pont{\gamma}{\A}{\C} +
  \pont{\gamma}{\C}{\B} < \pont{\gamma}{\A}{\gap} +
  \pont{\gamma}{\B}{\gap} \le \qmax + \qmax = 2 \cdot
  \qmax$. Consequently, we have that $\costN{\gamma}[A_{\{h, i \}}] =
  \costA{\gamma}[A_{\{h, i \}}]/\abs{A_{\{h, i \}}} \le 2 \cdot \qmax
  \abs{A_{\{ h, i \}}}/ \abs{A_{ \{ h, i \}}} = 2 \cdot \qmax$.
\end{proof}

\begin{lem}\label{lemma:approx:3}
  Let $S$ be a $k$-sequence, $\mystar$ a star of $S$ with center $c$, 
  $\secstar$ a $\mystar$-starsplitting and
  $A = \textsc{CompatibleAlign}(\secstar)$. Then,
  \begin{enumerate}
  \item[(i)] $\costA{\gamma}[A_{\{h, i\}}] \le \costA{\gamma}[A_{\{h,
      c\}}] + \costA{\gamma}[A_{\{c, i\}}]$\, when $\gamma \in
  \metricA$  and
  \item[(ii)] $\costN{\gamma}[A_{\{h, i\}}] \le 2 \cdot \Big(
    \costN{\gamma}[A_{\{h, c\}}] + \costN{\gamma}[A_{\{c, i\}}] \Big)$\, when $\gamma \in
  \metricN$,
  \end{enumerate}
  for each $h < i$, $h \not= c$, $i \not= c$.
\end{lem}
\begin{proof}
Consider $A = [s'_{1}, \ldots, s'_{k}]$, $h,i \in \NN$ and $J = \{j :
s'_{c}(j) \not= \gap \mbox{ and } s'_{h}(j) = s'_{i }(j) = \gap \}$. Then, 
\begin{align*}
  \costA{\gamma}[A_{\{h, i\}}] + \sum_{j \in J}
  (\pont{\gamma}{\gap}{s'_{c}(j)} + \pont{\gamma}{s'_{c}(j)}{\gap})
  &= \sum_ {j \not\in J} \pont{\gamma}{s'_{h}(j)}{s'_{i}(j)} + \sum_{j \in J} (\pont{\gamma}{\gap}{s'_{c}(j)} + \pont{\gamma}{s'_{c}(j)}{\gap}) \\
  &\le \sum_{j \not\in J} \left( \pont{\gamma}{s'_{h}(j)}{s'_{c}(j)} + \pont{\gamma}{s'_{c}(j)}{s'_{i}(j)} \right) + \sum_{j \in J} (\pont{\gamma}{\gap}{s'_{c}(j)} +
  \pont{\gamma}{s'_{c}(j)}{\gap}) \\
  &= \costA{\gamma}[A_{\{h, c\}}] + \costA{\gamma}[A_{\{c, i\}}]\,,
\end{align*}
where the inequality holds due to
Proposition~\ref{proposition:approx:2}.  Therefore,
\begin{align}
\costA{\gamma}[A_{\{h, i\}}] &\le \costA{\gamma}[A_{\{h, c\}}] +
\costA{ \gamma}[A_{\{c, i\}}] - \sum_{j \in
  X}(\pont{\gamma}{\gap}{s'_{c}(j)} +
\pont{\gamma}{s'_{c}(j)}{\gap})\,.
\label{exp:approx:5}
\end{align}
Since $\gamma \in \metricA$, we have that
$\pont{\gamma}{\gap}{s'_{c}(j)}, \pont{\gamma}{s'_{c}(j )}{\gap} >
0$. It follows from~(\ref{exp:approx:5}) that (\textit{i}) is proven.

For (\textit{ii}), observe that, by definition of $\metricA$, we have
that
\[
\qmax = \max_{\sigma \in \Sigma} \big\{ \pont{\gamma}{\sigma}{\gap},
\pont{\gamma}{\gap}{\sigma} \big\} = \max_{\sigma \in \Sigma} \big\{
\pont{\gamma}{\sigma}{\gap} \big\} \le 2\,
\pont{\gamma}{s'_{c}(j)}{\gap} = \pont{\gamma}{\gap}{s'_{c}(j)} +
\pont{\gamma}{s'_{c}(j)}{\gap} 
\]
for every $j$. Furthermore, following these statements, we have that
\begin{align}
\costN{\gamma}[A_{\{h, i\}}] &= \frac{\costA{\gamma}[A_{\{h,
      i\}}]}{|A_{\{h, i\}}|} \nonumber \\
&\le 
\frac{\costA{\gamma}[A_{\{h, i\}}] + 2 \cdot \qmax \abs{J}}{|A_{\{h , i\}}| +
\abs{J}} \; \le \; 2 \cdot \frac{\costA{\gamma}[A_{\{h, i\}}] + \qmax \abs{J}}{|A_{\{h, i\}}| + \abs{J}}\label{exp1:proposition:approx0}
\\
&\le 
2 \cdot  \frac{
\costA{\gamma}[A_{\{h, c\}}] +
\costA{\gamma}[A_{\{c, i\}}]
- \sum_{j \in J}(\pont{\gamma}{\gap}{s'_{c}(j)} + \pont{\gamma}{s'_{c}(j) }{\gap}) + \qmax \abs{J}} 
{|A_{\{h, i, c\}}| - \abs{J} + \abs{J}}
\label{exp1:proposition:approx1}\\
&\le 
2 \cdot \frac{\costA{\gamma}[A_{\{h, c\}}] + \costA{\gamma}[A_{\{c, i\}}] -
\qmax \abs{J}+ \qmax\abs{J}}
{|A_{\{h, i, c\}}| - \abs{J}+ \abs{J}} =
2\cdot \left( \frac{\costA{\gamma}[A_{\{h, c\}}]}{|[A_{\{h, i, c\}}]|}+
\frac{\costA{\gamma}[A_{\{c, i\}}]}{|[A_{\{h, i, c\}}]|}\right)
\label{exp1:proposition:approx2}\\
&\le 
2 \cdot
\left( \frac{\costA{\gamma}[A_{\{h, c\}}]}{|A_{\{h, c\}}|}+
\frac{\costA{\gamma}[A_{\{c, i\}}]}{|A_{\{c, i\}}|} \right) =
2 \cdot \Big( \costN{\gamma}[A_{\{h, c\}}]+\costN{\gamma}[A_{\{c, i\}}] \Big)\,, \label{exp1:proposition:approx55}
\end{align}
where the first inequality of~(\ref{exp1:proposition:approx0}) is a
consequence of Proposition~\ref{proposition:approx:2} and the second
inequality follows since every entry of $\gamma$ is nonnegative,
(\ref{exp1:proposition:approx1}) follows from~(\ref{exp:approx:5}) and
from $|A_{\{h, i\}}| = |A_{\{h, i, c\}}| -
\abs{J}$,~(\ref{exp1:proposition:approx2}) follows as a consequence of
$\gamma \in \metricN$, and (\ref{exp1:proposition:approx55})
follows as a consequence of $|A_{\{h, c\}}| \le |A_{\{h, i, c\}}|$ and
$|A_{\{c, i\}}| \le |A_{\{h, i, c\}}|$.
\end{proof}

Observe now that the running time of \textsc{CompatibleAlign} is
$O(k^{2}n)$.

\begin{algorithm}\caption{}\label{alg-aprox-SP}
\begin{algorithmic}[1]
\REQUIRE $k$-sequence $S$ \\ \ENSURE
$\generalcost[A]$ such that $A \in \mathcal{A}_S$ and\\
$\costSP{\gamma}[A] \le 6 \cdot \DistanceSP{\gamma}(S)$ if
$\generalcost = \costA{\gamma}$ and $\gamma \in \metricA$, and\\ $\costSPN{\gamma}{2}[A] \le 12
\cdot \DistanceNSP{\gamma}{2}(S)$ if $\generalcost = \costN{\gamma}$ and $\gamma \in \metricN$.

\medskip

\STATE Let $\mystar$ be a $\generalcost$-optimal star of $S$ with center $c$
\STATE Compute the $\mystar$-splitting $\secstar$
\STATE $A \gets \textsc{CompatibleAlign}(\secstar)$
\STATE \textbf{return} $\generalcost[A]$
\end{algorithmic}
\end{algorithm}

\begin{thm}
  Let $S$ be a $k$-sequence and $\gamma$ be a
  scoring matrix. Then,
  Algorithm~\ref{alg-aprox-SP} computes $\generalcost[A]$ correctly, 
  \begin{enumerate}
  \item[(i)] in $O(k^{2}n^{2})$-time such that $\costSP{\gamma}[A] \le
    6 \cdot \DistanceSP{\gamma}(S)$, if $\generalcost =
    \costA{\gamma}$ and $\gamma = \metricA$,
  \item[(ii)] in $O(k^{2}n^{3})$-time such that
    $\costSPN{\gamma}{2}[A] \le 12 \cdot \DistanceNSP{\gamma}{2}(S)$,
    if $\generalcost = \costN{\gamma}$ and $\gamma \in \metricN$\,,
  \end{enumerate}
  where $A$ is the alignment of $S$ computed by the algorithm.
\end{thm}
\begin{proof}
Clearly, the value returned by the Algorithm~\ref{alg-aprox-SP} is a
score of an alignment of $S$.

We show then that the approximation factor is as expected.
Let $c$ be a center of the stars $X$ and $Y$ found in the first two steps of Algorithm~\ref{alg-aprox-SP}.
Notice then that
\begin{align}
\sum_{h=1}^{k-1} \sum_{i = h+1}^{k}
\big(\generalcost [A_{ \{ h, c \} }] +
\generalcost [A_{ \{ c, i \} }] \Big) &= 
(k-1) \cdot \cost(\secstar)\label{exp:approx:2} \\
&\le 3 \cdot (k-1) \cdot \cost(D(\mystar))\label{exp:approx:3} \\
&\le 3 \cdot (k-1) \cdot \frac{2}{k} \cdot
\opt(S) = 6 \cdot \frac{k-1}{h} \cdot \opt(S) \le 6 \cdot \opt (S)\,,\label{exp:approx:4}
\end{align}
where the equality~\eqref{exp:approx:2} follows since
$c$ is the center of star  
$\secstar$ which is a
compatible with alignment $A$,~\eqref{exp:approx:3} follows from
Lemma~\ref{proposition:approx:1} and~\eqref{exp:approx:4} follows from
Lemma~\ref{lemma:approx:2}.

Suppose then that $\generalcost = \costA{\gamma}$ and $\gamma \in
\metricA$. Thus,
\begin{align*}
\costSP{\gamma}[A] &= 
\sum_{h=1}^{k-1} \sum_{i = h+1}^{k}
\costA{\gamma}[A_{\{h, i\}}] \\
&\le 
\sum_{h=1}^{k-1} \sum_{i = h+1}^{k}
\Big(\costA{\gamma}[A_{\{h, c\}}] +
\costA{\gamma} [A_{\{c, i\}}]\Big)
\le 6 \cdot \DistanceSP{\gamma}(S)\,,
\end{align*}
where the first inequality follows from Lemma~\ref{lemma:approx:3} and
the second follows from Equation~(\ref{exp:approx:4}).

Suppose now that $\generalcost = \costN{\gamma}$ and $\gamma \in \metricN$.
Thus, 
\begin{align*}
\costSPN{2}{\gamma}[A] &= 
\sum_{h=1}^{k-1} \sum_{i = h+1}^{k}
\costN{\gamma}[A_{\{h, i\}}] \\
&\le 
2 \cdot \sum_{h=1}^{k-1} \sum_{i = h+1}^{k}
\Big(\costN{\gamma}[A_{\{h, c\}}] +
\costN{\gamma}[A_{\{c, i\}}]\Big)
\\
&\le  2 \cdot 6 \cdot \DistanceNSP{2}{\gamma} (S) = 12 \cdot \DistanceNSP{2}{\gamma}(S)\,,
\end{align*}
where, similarly, the first inequality follows from
Lemma~\ref{lemma:approx:3} and the second follows from
Equation~(\ref{exp:approx:4}).

The time required to find an optimal $\generalcost$-star is the time
to compute the pairwise alignments of $S$, which is ${k \choose 2}
O(n^{2})$ if $\generalcost = \costA{\gamma}$ and it is ${k \choose 2}
O(n^{3})$ if $\generalcost = \costN{\gamma}$. Additionally, we have to
consider the time to determine the optimal star, which is $O (k^{2})$,
implying that the time required to compute line 1 of
Algorithm~\ref{alg-aprox-SP} is ${k \choose 2} O(n^{2}) + O(k^2) =
O(k^{2} n^ {2})$ if $\generalcost = \costA{\gamma}$ and ${k \choose 2}
O(n^{3}) + O(k^2) = O(k^{2} n^ {3})$ if $\generalcost =
\costN{\gamma}$. The time spent to compute lines 2 and 3 are $O(kn)$
and $O(k^{2}n)$, respectively, and to compute line 4 is $O(k^{3}n)$,
since we have to compute the score of ${k \choose 2} = O(k^{2})$
pairwise alignments of length $O(kn)$. Therefore, the total time spent
by the algorithm is $O(k^{2 } n^{2} + k^{3} n)$ if $\generalcost =
\costA{\gamma}$ and $O(k^{2 } n^{3} + k^{3} n)$ if $\generalcost =
\costN{\gamma}$.
\end{proof}

\section{Conclusion and future work}\label{sec:concl}

We presented and discussed multiple aspects of normalized multiple
sequence alignment (NMSA). We defined three new criteria for computing
normalized scores when aligning multiple sequences, showing the
NP-hardness and exact algorithms for solving the $\NMSA\text{-}z$ given each
criterion $z = 1, 2, 3$. In addition, we adapted an existing
2-approximation algorithm for $\MSA$ when the scoring matrix $\gamma$
is in the classical class $\metricC$, leading to a 6-approximation
algorithm for $\MSA$ when $\gamma$ is in the broader class $\metricA
\supseteq \metricC$ and to a 12-approximation for $\NMSA\text{-}2$ when
$\gamma$ is in $\metricN \subseteq \metricA$, a slightly more
restricted class such that the cost of deletion for any symbol is at
most twice the cost for any other.
We summarize these contributions in Table=ref{tableconclusion}.

\begin{table}[htpb]
\begin{center}
\begin{tabular}{c|c|c}
\textbf{problems} & \textbf{time of exact algorithms} & 
\textbf{time of approximation algorithm}\\
\hline \hline
$\MSA$ & --- & $O (k^{2} n^{2} + k^{3} n)$\\
\hline
$\NMSA\text{-}1$ & $O (2^{k}k^{3} (n+1)^{k+1})$ &
--- \\
\hline
$\NMSA\text{-}2$ &
$O \left(\left( 1 + \frac{1}{2n+1} \right)^{k} (2n + 1)^{k^{2}}k^{2}\right)$ &
$O (k^{2} n^{2} + k^{3} n)$
\\
\hline
$\NMSA\text{-}3$ & $O (2^{k}  k^{4} (n+1)^{k+1})$ & ---\\
\hline \hline
\end{tabular}
\end{center}
\caption{We are considering a $k$-sequence where each sequence has
  maximum length of $n$;
  all the problem decision version are NP-complete.}
\label{tabelaCarlinhosVarias}
\end{table}

This work is an effort to expand the boundaries of multiple sequence
alignment algorithms towards normalization, an unexplored domain that
can produce results with higher accuracy in some applications. In
future work, we will implement our algorithms in order to verify how
large are the sequences our algorithms are able to handle. Also, we
plan to perform practical experiments, measuring how well alignments
provided by our algorithms and other MSA algorithms agree with multiple
alignment benchmarks. In addition, we intend to measure the accuracy
of phylogenetic tree reconstruction based on our alignments for
simulated and real genomes. Finally, we will work on heuristics and
parallel versions of our algorithms in order to faster process large
datasets.

% Bibliography
%-----------------------------------------------------------------
% \bibliographystyle{abbrv}
\bibliographystyle{alpha}
\bibliography{bibliography}

\end{document}